\documentclass[aps,prd,notitlepage,nofootinbib,superscriptaddress,groupedaddress]{revtex4-2}

\usepackage[latin9]{inputenc}
\usepackage{mathrsfs}
\usepackage{bm}
\usepackage{amsmath}
\usepackage{amsthm}
\usepackage{amssymb}
\usepackage{graphicx}
\usepackage{esint}
\usepackage{hyperref}
\usepackage{tabularx}

\theoremstyle{plain}

\theoremstyle{plain}

\ifx\proof\undefined
\newenvironment{proof}[1][\protect\proofname]{\par
	\normalfont\topsep6\p@\@plus6\p@\relax
	\trivlist
	\itemindent\parindent
	\item[\hskip\labelsep\scshape #1]\ignorespaces
}{%
	\endtrivlist\@endpefalse
}
\providecommand{\proofname}{Proof}
\fi

\usepackage[table ]{ xcolor}

\usepackage{amssymb}
\usepackage{graphicx,color}

\usepackage{amsmath,amsthm}

\usepackage{mathrsfs}

\usepackage{bm}

\def\beq{\begin{equation}}
\def\eeq{\end{equation}}
\def\bi{\begin{itemize}}
\def\ei{\end{itemize}}
	\def\ba{\begin{array}}
	\def\ea{\end{array}}
	\def\bfig{\begin{figure}}
	\def\efig{\end{figure}}

	\def\C{\mathbb{C}}
	\def\R{\mathbb{R}}
	\def\Z{\mathbb{Z}}

	\newtheorem{theorem}{Theorem}[section]

	\newtheorem{lemma}[theorem]{Lemma}

	\newcommand{\bA}{{\bar{A}}}

	\newcommand{\Slc}{\mathrm{SL}(2,\mathbb{C})}

	\newcommand{\Su}{\mathrm{SU}(2)}

	\def\be{\begin{eqnarray}}
	\def\ee{\end{eqnarray}}

	\newcommand{\cc}{\mathcal C}

	\newcommand{\cf}{\mathcal F}
	\newcommand{\cg}{\mathcal G}
	\newcommand{\ch}{\mathcal H}
	
	\newcommand{\cj}{\mathcal J}
	\newcommand{\ck}{\mathcal K}
	
	\newcommand{\cm}{\mathcal M}

	\newcommand{\calp}{\mathcal P}

	\newcommand{\cs}{\mathcal S}
	\newcommand{\ct}{\mathcal T}
	\newcommand{\cu}{\mathcal U}
	\newcommand{\cv}{\mathcal V}

	\newcommand{\sa}{\mathscr{A}}
	\newcommand{\sm}{\mathscr{M}}

	\newcommand{\ff}{\mathfrak{f}}

	\newcommand{\fl}{\mathfrak{l}}  
	  
	\newcommand{\fn}{\mathfrak{n}}

	  \newcommand{\Fr}{\mathfrak{R}}
	\newcommand{\fs}{\mathfrak{s}}  \newcommand{\Fs}{\mathfrak{S}}

	\renewcommand{\a}{\alpha}
	\renewcommand{\b}{\beta}
	\newcommand{\g}{\gamma}
	\newcommand{\G}{\Gamma}
	
	\newcommand{\eps}{\varepsilon}

	\newcommand{\sig}{\sigma}
	\newcommand{\Sig}{\Sigma}
	\renewcommand{\l}{\lambda}
	\renewcommand{\L }{\Lambda}
	\renewcommand{\o}{\omega}
	\renewcommand{\O}{\Omega}
	\renewcommand{\t}{\tau}

	\newcommand{\rmd}{\mathrm d}

	\newcommand{\lt}{\left}
	\newcommand{\rt}{\right}

	\newcommand{\lag}{\left\langle}
	\newcommand{\rag}{\right\rangle}

	\newcommand{\sn}{\mathscr{N}}

	\newcommand{\re}{\mathrm{Re}}

	\newcommand{\tr}{\mathrm{Tr}}

\begin{document}

\title{On the summation and triangulation independence of Lorentzian spinfoam amplitudes for all LQG}

\author{Muxin Han}
\email{hanm(At)fau.edu}
\affiliation{Department of Physics, Florida Atlantic University, 777 Glades Road, Boca Raton, FL 33431, USA}
\affiliation{Institut f\"ur Quantengravitation, Universit\"at Erlangen-N\"urnberg, Staudtstr. 7/B2, 91058 Erlangen, Germany}

\begin{abstract}

This paper investigates the fundamental issue of triangulation dependence in spinfoam quantum gravity. It introduces a novel framework, named spinfoam stack, to systematically sum spinfoam amplitudes over an infinite class of 2-complexes. These complexes are generated by stacking an arbitrary number of faces upon a simpler root complex. The central result is obtained by analyzing the amplitude of spinfoam stack in the limit where an upper cut-off on the area of internal faces is taken to infinity. In this limit, the amplitude as an integral localizes via a stationary phase mechanism onto a critical manifold. This manifold is shown to be the space of SU(2) flat connections on the underlying complex. This localization effectively reduces the bulk dynamics from a theory of quantum geometry to a topological theory akin to SU(2) BF theory. For spinfoams on topologically trivial manifolds, this result has a powerful consequence: the spinfoam stack amplitude factorizes into a triangulation-dependent normalization factor and a finite part that depends only on the boundary data. Renormalizing the amplitude yields a finite result that is manifestly independent of the bulk structure of the 2-complex. This provides a concrete realization of triangulation independence in a well-defined limit, suggesting the possibility of existing a non-trivial fixed point of quantum gravity within the spinfoam formalism.
\end{abstract}

\maketitle

\tableofcontents

\section{Introduction}

Loop Quantum Gravity (LQG) offers a compelling framework for a non-perturbative and background-independent theory of quantum gravity \cite{rovelli2004quantum, ashtekar2004background, thiemann2008modern, Perez2012,han2007fundamental}. In the canonical formulation of LQG, the fundamental quantum geometry excitations of space are described by spin-network states \cite{rovelli1988knot, rovelli1995discreteness}. These are graphs embedded in a spatial manifold, with links colored by representations of the SU(2) group and nodes by intertwiners. These states form an orthonormal basis for the kinematical Hilbert space $\mathcal{H}_{\text{Kin}}$ \cite{Ashtekar:1994wa, uniqueness}, and geometric operators for area and volume are found to have discrete spectra \cite{rovelli1995discreteness, ashtekar1997quantum, ashtekar1997quantumII}. This discreteness of geometry at the Planck scale is a cornerstone prediction of the theory.

While the canonical approach provides a detailed picture of quantum geometry at a moment in time, a complete theory must also describe its evolution. The covariant, or path integral, formulation of LQG, known as spinfoam theory, aims to provide this dynamical description \cite{Reisenberger:1996pu, Baez2000, Perez2012, rovelli2014covariant}. In this framework, the transition amplitude between initial and final spin-network states is computed as a "sum over histories" of quantum geometries. A single history, or spinfoam, is based on a 2-complex, which can be visualized as the spacetime evolution of a spin-network graph. The links of the spin-network evolve into faces of the 2-complex, and the nodes evolve into edges. The evolution carries the spins/intertwiners from links/nodes to edges/faces. A spinfoam is constructed by assigning spins to the faces and intertwiners to the edges. A spinfoam model associates a complex weight for each spinfoam, a spinfoam amplitude based on a 2-complex sums the weights over all assignments of spins and intertwiners on the complex.

As popular spinfoam model in this program, the Engle-Pereira-Rovelli-Livine (EPRL) model \cite{EPRL} provides a concrete prescription for the spinfoam amplitude for Lorentzian gravity. The model describes the dynamics of arbitrary states in the Hilbert space $\ch_{\rm Kin}$ by the Kaminski-Kisielowski-Lewandowski (KKL) formalism \cite{KKL,generalize}. There has been extensive investigations on the behavior of the model in the semiclassical regime of large spins (see e.g \cite{semiclassical, HZ,Han:2023cen,Han:2021kll,claudio,propagator3}). The EPRL spinfoam amplitude has been shown to relate to the Regge calculus in this regime, providing strong evidence that the model correctly captures the dynamics of discrete general relativity.

Despite these successes, a major unresolved issue in spinfoam theory is the dependence of the amplitudes on the underlying 2-complex. Physical predictions should be independent of this auxiliary discretization, which is chosen for computational convenience. The standard proposal to address this is to sum the amplitudes over all possible 2-complexes compatible with the given boundary conditions \cite{Reisenberger:1996pu}. This is analogous to summing over all Feynman diagrams in quantum field theory or over all triangulations in dynamical triangulation models. However, this sum over complexes has been notoriously difficult to define and control. Group Field Theory (GFT) provides a formal framework for organizing this sum as a perturbative expansion, where the GFT Feynman diagrams correspond to the 2-complexes \cite{Oriti2007}. While GFT has yielded significant insights, how to compute the sum remains to be a difficult problem. A direct, non-perturbative understanding of the sum over complexes within the spinfoam formalism itself is highly desirable, and is closely connected to renormalization in spinfoam quantum gravity \cite{Dittrich:2014ala, Bahr:2017klw}.

This paper introduces a novel framework to systematically perform a sum over complexes and investigates its properties in a specific, physically motivated limit. We propose the concept of a ``spinfoam stack'', which organizes the sum over an infinite class of 2-complexes. The construction begins with a simple "root complex" $\mathcal{K}$, where each loop in the 1-skeleton bounds at most one face. A family of more intricate complexes is then generated by ``stacking'' an arbitrary number of faces, $p_f$, upon each face $f$ of the root complex. The stack amplitude (based on the root complex $\ck$) is then defined as a sum over these stacked complexes, weighted by coupling constants $\lambda_f$ associated with each root face $f$. The construction of spinfoam stack is motivated by the structure of the LQG Hilbert space, where generic states can be represented as linear combinations of ``spin-network stacks'' that are superpositions of spin-networks on graphs with varying link multiplicities. The spinfoam stack naturally describes the covariant evolution of such states. Finally, the complete amplitude sums the stack amplitudes over different root complexes sharing the same boundary.

To render the sum over an infinite number of spin configurations well-defined, we introduce a regularization by imposing an upper bound, $A_f$, on the total LQG area given by the spins on the faces stacked upon each root face $f$. This cut-off is inspired by physical considerations, such as the presence of a cosmological horizon, which sets a maximal observable area. The central focus of this work is to analyze the behavior of the stack amplitude in the limit where the area cut-offs of internal faces are taken to be large, $A_f \to \infty$ (keeping boundary areas fixed). 


Our main result is that in this large-cutoff limit, the stack amplitude of the Lorentzian EPRL spinfoam model undergo a remarkable simplification. The path integral over the $\mathrm{SL}(2,\mathbb{C})$ group variables, representing the stack amplitude, localizes via a stationary phase mechanism onto a critical manifold $\cc_{\rm int}$. For some details, the summation over the stack multiplicities $\{p_f\}$ can be computed by Laplace transform method, in a similar way as state-counting in computing LQG black hole entropy (see e.g. \cite{Agullo:2010zz,BarberoG:2008dee,Engle2011,lqgee1}). The large-cutoff limit enforces a sharp constraint on the amplitude. The sum is dominated by configurations that maximizes the real part of an effective action $S$ and satisfies $\partial S=0$. These configurations are the critical points of the effective action and form the critical manifold $\cc_{\rm int}$.

We demonstrate that the equations defining this critical manifold are precisely the conditions for reducing $\Slc$ group variables to SU(2) holonomies on the root complex, such that the loop holonomy around every face is trivial. Moreover, the critical manifold $\cc_{\rm int}$ quotient by the on-shell SU(2) gauge group $\cg_{\rm int}$ is shown to be isomorphic to the moduli space of SU(2) flat connections on the root complex. It is also isomorphic to the moduli space of SU(2) flat connections on the 4-manifold $\mathcal{M}_4$ where the root 2-complex $\ck$ is embedded, if $\ck$ is sufficiently refined.  Localizing the spinfoam path integral onto $\cc_{\rm int}$ effectively reduces the dynamics in the bulk from a theory of Lorentzian quantum geometry to a topological theory akin to SU(2) BF theory \cite{Ooguri:1992eb} (but with a different path integral measure). This result suggests a possible connection between the high-energy regime of spinfoam quantum gravity and topological quantum field theory.

This general result has a particularly powerful and concrete consequence for spinfoams on manifolds with trivial topology, where the root complex has a trivial fundamental group, $\pi_1(|\ck|) = \{1\}$, then the moduli space of flat connections trivially consists of a single point corresponding to the trivial connection. In this case, the entire path integral localizes to a single configuration and the on-shell gauge freedom. We show that the stack amplitude $\mathscr{A}_{\mathcal{K}}$ factorizes into a product
\begin{equation}
\mathscr{A}_{\mathcal{K}} = \mathscr{N}_{\mathcal{K}}  \mathscr{A}_{\Gamma,\fs}.
\end{equation}
Here, $\mathscr{N}_{\mathcal{K}}$ is a normalization factor that arises from the Gaussian integration over the fluctuations around the critical point. This factor depends on the bulk structure of the root complex $\mathcal{K}$. In the large-cutoff limit, this factor diverges. Crucially, the second factor $\mathscr{A}_{\Gamma,\fs}$ is finite and depends only on boundary data, independent of the choice of the bulk structure in $\mathcal{K}$. By defining a renormalized amplitude $\mathscr{A}_{\text{ren}} = \mathscr{A}_{\mathcal{K}} / \mathscr{N}_{\mathcal{K}}=\sa_{\G,\fs}$, we obtain a finite result that is manifestly independent of the bulk triangulation. This provides a concrete and powerful realization of triangulation independence in a well-defined limit of the Lorentzian spinfoam model. In addition, the further summation of $\sa_\ck$ over root complexes becomes simplified by this triangulation independence.

The physical interpretation of these results points towards a rich phase structure for spinfoam quantum gravity. In the large area cutoff limit, the theory relates to a topological, scale-invariant phase. The triangulation independence of the renormalized amplitude is a direct manifestation of this scale invariance; refining the triangulation corresponds to probing smaller scales, and the invariance of $\mathscr{A}_{\Gamma,\fs}$ suggests that the theory might have reached a non-trivial fixed point. This picture might share conceptual similarities with the Asymptotic Safety scenario for quantum gravity, where the theory is predicted to have a UV fixed point \cite{Eichhorn:2018phj}. Our result suggests a realization of the scenario within the spinfoam framework.

Furthermore, our framework can connect to the established semiclassical results of spinfoam gravity, which correspond to the theory's infrared (IR) regime. We argue that for finite area cutoffs $A_f$ and in a regime where the Barbero-Immirzi parameter is small, the spinfoam stack amplitude is no longer dominated by the above topological theory. Instead, the sum is dominated by the amplitude on the root complex where all $p_f=1$. This occurs because the coupling constants $\lambda_f$ should be small for small $\g$, suggested by the recent result of entanglement entropy from spinfoam \cite{spinfoamstack}, so $\lambda_f$ suppresses the amplitude on the complexes with $p_f>1$. This ensures that our framework connects to established semiclassical results based only on root complexes, including the correspondence between the spinfoam amplitude and the Regge calculus. 

The regime of these semiclassical results corresponds to both boundary spins and the internal spin cut-offs being uniformly large but finite. In contrast, the regime studied in this paper is the infinite internal cut-off limit $A_f\to\infty$ while keeping boundary state fixed. The semiclassical results and the results here should correspond to the behavior of the theory at two different regimes.

The organization of this paper is as follows. In Section \ref{Stacking spin-networks}, we introduce the concept of spin-network stacks as states in the LQG Hilbert space. In Section \ref{Spinfoam stack}, we extend this concept to the covariant picture, defining the spinfoam stack and the corresponding stack amplitude for the generalized EPRL model. Section \ref{Localization to the space of flat connections} contains our main analytical results. We use a Laplace transform and stationary phase methods to analyze the stack amplitude in the large area cutoff limit, demonstrating its localization to the space of SU(2) flat connections. Section \ref{Trivial topology and triangulation independence} proves the triangulation independence of the renormalized amplitude for topologically trivial manifolds. In Section \ref{Discussion}, we discuss the physical implications of our results, including the interpretation of the large-cutoff limit as a UV fixed point and the connection to the IR regime of the theory. Section \ref{Some explicit computations} provides some explicit computations for the case of trivial topology, including an explicit parametrization of the group variables and a proof of the non-degeneracy of the Hessian matrix governing the fluctuations around the critical manifold.

\section{Spin-network stack}\label{Stacking spin-networks}

Let $\G$ be a closed, oriented graph. The link multiplicity between any two nodes in $\G$ is the number of links connecting them. A spin-network state on $\G$ is defined by coloring each oriented link $\fl$ with a spin $j=k/2$ (with $k\in\Z_+$) and coloring each node $\fn$ with a normalized intertwiner $I_\fn$. We first assume that the multiplicities in $\G$ is less or equal to 1, in other words, any two nodes are connected at most by a single link if they are connected. Let us focus on a link $\fl$ that connects a source node $\fn_1=s(\fl)$ to a target node $\fn_2=t(\fl)$, the state can be expressed by
\be
\cdots\lt(I_{\fn_2}\rt)^{k;\cdots}_{m;\cdots} \Pi^{k}_{m,n}\lt(H_\fl\rt)\lt(I_{\fn_1}\rt)^{k;\cdots}_{n;\cdots}\cdots,\label{spinnetworkstate}
\ee
Here the normalized Wigner $D$-function of the SU(2) holonomy $H_\fl$ is denoted by $\Pi^{k}_{m,n}(H_\fl)=\sqrt{d_k} D^{k}_{m,n}(H_\fl)$, where $d_k=k+1$. The ellipses $\cdots$ represent quantities associated with links other than $\fl$. Contractions of the magnetic indices $m,n$ occur between the intertwiners and the Wigner $D$-functions.

A family of spin-network states can be generated from $\G$ by increasing the link multiplicities between any two neighboring nodes, such as $\fn_1$ and $\fn_2$ (two nodes are neighboring if the link multiplicity between these two nodes is one in $\G$). A typical state in this family takes the following form, when focusing on the links stacked upon $\fl$:
\be
\cdots\lt(I_{\fn_2}\rt)^{k_1\cdots k_p;\cdots}_{m_1\cdots m_p;\cdots}\prod_{i=1}^p \Pi^{k_i}_{m_i,n_i}\lt(H_{\fl(i)}\rt)\lt(I_{\fn_1}\rt)^{k_1\cdots k_p;\cdots}_{n_1\cdots n_p;\cdots}\cdots.
\ee
In this state, a total of $p$ links, denoted $\fl(i)$ for $i=1,\cdots,p$, connect nodes $\fn_1$ and $\fn_2$. An SU(2) holonomy $H_{\fl(i)}$ and a spin $k_i/2$ are carried by each individual link. The intertwiners $I_{\fn_{1}}$ and $I_{\fn_{2}}$ become higher-valent to handle the increased number of connections.

We consider a general superposition of the spin-networks in this family. The superposition sums over in the link multiplicities $p$, the collection of spins $\vec{k}=(k_1,\cdots,k_p)$, and the intertwiners, as depicted in FIG.\ref{sn_stack}. In order that the resulting state is normalizable, the summation is truncated by imposing a constraint: for an arbitrary $A_\fl>0$, the total LQG area contributed from the $p$ links between a pair of neighboring nodes is not permitted to exceed the cut-off value $4\pi\g\ell_P^2 A_\fl$. This constraint is imposed to every pair of neighboring nodes and to each state in the superposition. The resulting states is written as:
\be
\Psi_{\G,\vec{A}}\lt(\vec{H}\rt)&=&\sum_{\vec \mu}\sum_{\{I_\fn\}_\fn} C_{\vec \mu,\{I_\fn\}_\fn}\prod_{\fl\subset\G} \Theta\left(A_{\fl}-\alpha_{p_\fl,\vec{k}(\fl)}\right) \tr\lt( \bigotimes_{\fn\in\G} I_\fn \cdot \bigotimes_{\fl\subset\G}\lt[ \bigotimes_{i=1}^{p_\fl} \Pi^{k_i(\fl)}\lt(H_{\fl(i)}\rt)\rt]\rt),\label{spinnetworkstack}\\
\vec\mu&=&\lt(\lt\{p_\fl,\vec{k}(\fl)\rt\}\rt)_{\fl},\qquad 
\alpha_{p,\vec{k}}=\sum_{i=1}^{p}\sqrt{k_{i}(k_{i}+2)}.
\ee
We have use $\alpha_{p,\vec{k}}$ to denote the area spectrum. In the expression \eqref{spinnetworkstack}, $p_\fl \in \Z_+$ represents the link multiplicity associated to the original link $\fl$, while $k_i(\fl) \in \Z_+$ are their corresponding spins. The constraint is imposed by the Heaviside step function, $\Theta(x)$, defined as $\Theta(x)=1$ for $x\geq 0$ and $\Theta(x)=0$ for $x< 0$. The complete contraction of all magnetic indices is indicated by the trace $\tr$. The expression \eqref{spinnetworkstack} represents the state as a cylindrical function of holonomies. An alternative, representation-independent form of the state is given by:
 \be
|\Psi_{\G,\vec{A}}\rangle =\sum_{\vec{\mu}}\sum_{\{I_\fn\}_\fn}C_{\vec{\mu},\{I_\fn\}_{\fn}}\prod_{\fl\subset\G} \Theta\left(A_{\fl}-\alpha_{p_\fl,\vec{k}(\fl)}\right)\bigotimes_{\fn\in\G} |I_{\fn}\rangle.\label{stackcoarsegr}
\ee
This state $\Psi_{\G,\vec{A}}$ with general coefficients $C_{\vec{\mu},\{I_\fn\}_{\fn}}$ is termed a \emph{spin-network stack}. The original graph $\G$ upon which it is built is termed the \emph{root graph}. A spin-network state is a special case of spin-network stack, by setting all $C_{\vec{\mu},\{I_\fn\}_{\fn}}$ to vanish except one.

In the LQG Hilbert space $\ch_{\rm Kin}$ that includes all graphs, densely many state can be represented as a linear combination of spin-network stacks (with some $C_{\vec{\mu},\{I_\fn\}_{\fn}}$ and $\vec A$) based on different root graphs. If we denote by $\ch_{\G,\rm{st}}$ the Hilbert space of all spin-network stacks on the root graph $\G$ (with arbitrarily large cut-offs), the LQG Hilbert space $\ch_{\rm Kin}$ can be decomposed into
\be
\ch_{\rm Kin}=\bigoplus_{\G}\ch_{\G,\rm{st}}.
\ee
In contrast to the standard spin-network decomposition of $\ch_{\rm Kin}$, the direct sum is only over root graphs whose link multiplicities are not greater than one.

The graphs in the stack share the same set of nodes as the root graph $\G$, and the links are stacked upon links in $\G$, so we can use the links and nodes in the root graph to label the quantities in the spin-network stack \eqref{stackcoarsegr}.

\begin{figure}[t]
\centering
\includegraphics[width=0.5\textwidth]{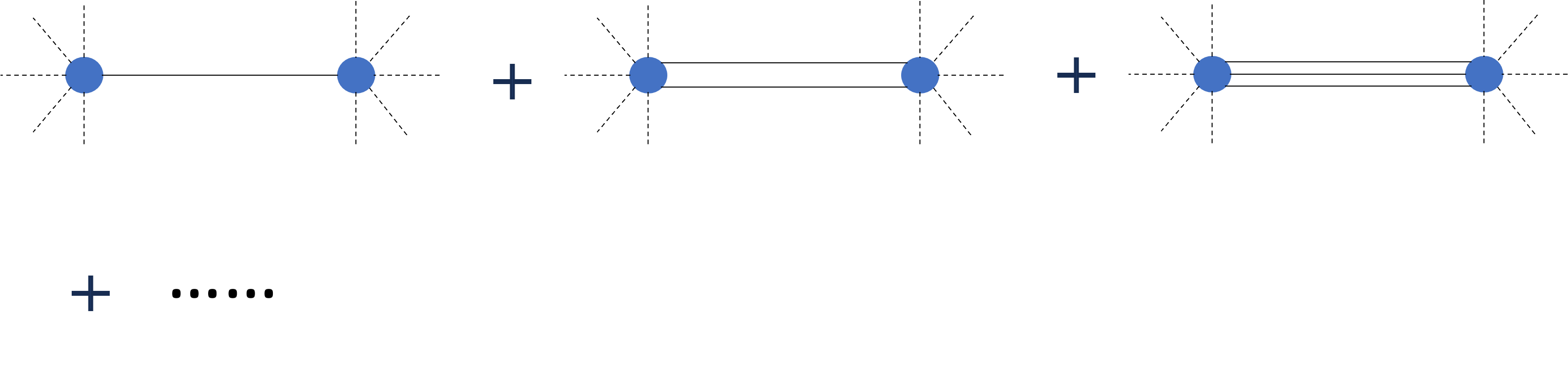}
\caption{The spin-network stack.}
\label{sn_stack}
\end{figure}

The spin-network state on the root graph $\G$ can be interpreted as the quantum geometry of a cellular decomposition of a spatial slice $\Sig$, and each intertwiner $I_\fn$ quantizes the geometry of a polyhedron with $M_{\fn}$ \emph{flat} faces, where $M_{\fn}$ is the valence of $\fn$ in the root graph $\G$ \cite{shape}. The spin-network on a generic graph in the stack corresponds to the same cellular decomposition, and the intertwiner $I_\fn$ still quantizes the polyhedron geometry with $M_{\fn}$ faces, whereas the \emph{curved} faces are allowed. Indeed, the links multiplicity $p_\fl>1$ associated to $\fl$ indicates that the polyhedron's face is discretized into $p_\fl$ small faces, while the 3d normals of the small faces are generally not parallel \cite{Han:2019emp}. The spin-network stack describes a quantum superposition of geometries with arbitrarily discretized curved faces.

\section{Spinfoam stack}\label{Spinfoam stack}

Given that spin-network stacks are well-defined states in the LQG Hilbert space and have interesting interpretations, the covariant dynamics of spin-network stacks must be taken into account in spinfoam theory.

A spinfoam is a covariant history of spin-network. Given a spin-network in 3d, the links $\fl$ and nodes $\fn$ evolve and become the faces $f$ and edges $e$ in (3+1) dimensions. A spinfoam is given by these faces and edges assigned respectively spins $j_f = k_f/2$ and intertwiners $I_e$ the same as the ones on links and nodes of the initial spin-network. Conversely, given a spinfoam, any spatial cross-section gives a spin-network. 

The faces and edges form a 2-complex underlying the spinfoam. The spinfoam amplitude, which defines a wave function of boundary spins and intertwiners, is defined on this chosen 2-complex. Consequently, most current investigations of spinfoams rely on a fixed 2-complex, leading to results that depend on this choice. However, a complete spinfoam formulation should yield predictions that are independent of the choice of 2-complexes. To achieve this, it is proposed that the amplitude should be summed over all possible 2-complexes. This approach is also motivated by the dynamics of LQG, as describing the evolution of generic LQG states--which are superpositions over different graphs--requires spinfoam amplitudes that are themselves a sum over various complexes.

Summing spinfoams over complexes motivates us to extend the concept of a stack to the spacetime picture. A spinfoam stack is a sum of spinfoams over a family of 2-complexes with the defining property that its intersection with any spatial slice yields a spin-network stack. The amplitude for a spinfoam stack, which we term the \emph{stack amplitude}, is the sum of the spinfoam amplitudes over all 2-complexes in the family.

The spinfoam faces and edges evolves respectively from spin-network links and nodes, so just as a spin-network stack is built by ``stacking'' links upon a root graph $\G$, a spinfoam stack is built by ``stacking'' faces upon a foundational root 2-complex, denoted $\ck$ (see FIG. \ref{sf_stack}). For any 2-complex, the face multiplicity of a closed loop in the 1-skeleton is the number of faces bounded by the loop. Any 2-complex is qualified to be a root complex if all face multiplicities are equal to one. Given a root complex $\ck$, a family $F(\ck)$ of 2-complexes can be generated from the root complex by arbitrarily increasing the face multiplicities. Given any root complex $\ck$, the stack amplitude $\sa_{{\cal K}}$ depending on $\ck$ is a sum of spinfoam amplitudes over the complexes in the family $F(\ck)$. The complete spinfoam amplitude is a sum of stack amplitudes $\sa_{{\cal K}}$ over root complexes sharing the same boundary.

The spinfoams are generally built on non-simplicial 2-complexes, so their amplitudes are constructed using the KKL formalism \cite{KKL,generalize}. Furthermore, in order to organize the sum over complexes, we define the coupling constant: $\l_f>0$ associated to each root face $f$. In the stack amplitude $\sa_{{\cal K}}$, the spinfoam amplitude on each complex in $F(\ck)$ is weighted by $\prod_{f\subset \ck}\l_f^{p_f}$, where $p_f$ is the face multiplicity at the root face $f$. The complete amplitude denoted by $\sa$ sums the stack amplitudes $\sa_{{\cal K}}$, each of which may be weighted by a coefficient $c_\ck\in\C$. In summary,
\be
\sa=\sum_{\ck}c_\ck\sa_\ck,\qquad \sa_\ck=\sum_{\{p_f\in\Z_+\}_{f\subset\ck}}\prod_{f\subset \ck}\l_f^{p_f}\sa \lt(\ck,\{p_f\}_{f\subset\ck}\rt),\label{sa6}
\ee
where $\sa \lt(\ck,\{p_f\}_{f\subset\ck}\rt)$ is the generalized EPRL spinfoam amplitude on the 2-complex in the family $F(\ck)$ with the face multiplicity $p_f$ at each root face $f$. 


The sum over complexes in the stack amplitude $\sa_\ck$ is compatible with the inner product on $\ch_{\rm Kin}$, so that $\sa_\ck$ is invariant under cut and gluing, in particular, $\sa_\ck$ can be expressed as gluing vertex amplitudes, as we will see in Section \ref{Stacking spinfoams}.




For any root face $f$, its dual face is endowed with the area $\a_{p_f,\vec{k}}$ in $\sa \lt(\ck,\{p_f\}_{f\subset\ck}\rt)$, where $\vec{k}=(k_1,\cdots,k_{p_f})$. We impose a constraint on the maximal area: $\a_{p_f,\vec{k}}\leq A_f$, similar to the spin-network stack. As a result, the sum $\sum_{p_f}$ and the sum over spins $\sa \lt(\ck,\{p_f\}_{f\subset\ck}\rt)$ are regularized to finite sums by the cut-off $A_f$. 

In the sum over root complexes in \eqref{sa6}, it turns out that our main result will only be sensitive to a finite number of degrees of freedom in the choice of coefficients $c_\ck$. It is because $\sa_\ck$ becomes independent of the bulk structure of $\ck$ in the limit of large cut-offs, as far as $\ck$ has a trivial fundamental group.

\begin{figure}[t]
\centering
\includegraphics[width=0.6\textwidth]{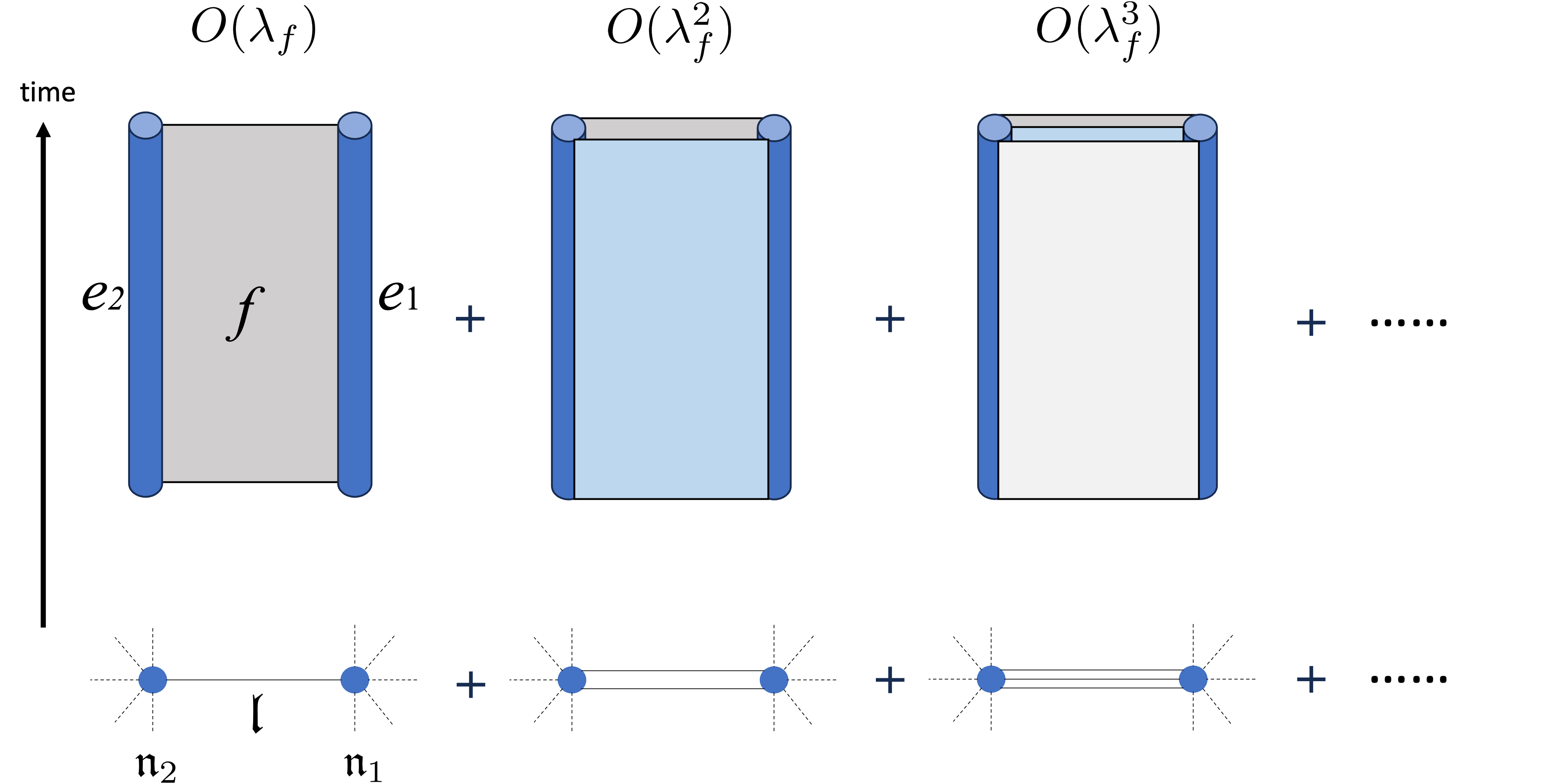}
\caption{The spin-network stack evolves to the spinfoam stack: The spin-network link $\fl$ evolves to the spinfoam face $f$. the spin-network nodes $\fn_1,\fn_2$ evolve to the spinfoam edges $e_1,e_2$. The faces evolves from the dashed links are not shown on this figure. The leftmost complex is the root complex. The power of coupling constant $\l_f$ counts the number of stacked faces.}
\label{sf_stack}
\end{figure}

\subsection{Stack vertex amplitude}\label{Stack vertex amplitude}

Consider a root complex that contains a single vertex $v$ connecting to a number $n_v$ of edges $e$. We generally assume $n_v\geq 5$, and the root complex is dual to a 4-simplex when $n_v=5$. The root complex contains a number of faces $f$, and each $f$ is bounded by a loop consisting of a pair of edges $e, e'$ and a boundary link $\fl_{f;e,e'}$. The face multiplicity equals one in the root complex. Following the discussion above \eqref{sa6}, the stack amplitude associated to the vertex $v$ is given by the following function of SU(2) holonomies $\vec{H}=\{H_{ef}^{(i)}\}$
\be
\sa_v\lt(\vec{A},\vec{H},\vec{\l}\rt)&=&\int  \prod_{e\text{ at }v}\rmd g_{ve}\, \delta(g_{ve_0})\prod_{f\text{ at }v}\o_f\lt(A_f;\{g_{ve}\},\vec{H},\l_{vf}\rt),\label{Av1}\\
\o_f&=&\sum_{p_f=1}^{\infty}\sum_{k_1,\cdots, k_{p_f}\in\Z_+}\prod_{i=1}^{p_f} \zeta^{(f)}_{k_i}\lt(\{g_{ve}\},\vec{H},\l_{vf}\rt)\Theta\left(A_{f}-\alpha_{p_f,\vec{k}}\right),\\
\zeta^{(f)}_{k_i}&=&\l_{vf} d_{k_i}\tr_{(k_i,\rho_i)}\lt[\lt(P_{k_i}g_{ve}^{-1}g_{ve'}P_{k_i}\rt) H_{e' f}^{(i)}H_{ef}^{(i)-1}\rt],\qquad \rho_i=\gamma (k_i+2),\label{Avstackampliude}
\ee
where $\l_{vf}$ and $A_f$ are the coupling constant and cut-off associated to the stacked faces at $f$. The amplitude is given by the Haar integrals $\int\rmd g_{ve}$ over $\Slc$ variables $g_{ve}$. Choosing arbitrarily an edges $e_0$, the delta function $\delta(g_{ve_0})$ fixes the noncompact gauge symmetry $g_{ve} \to x_v g_{ve}$, $x_v\in\Slc$ to make $\sa_v$ finite \cite{finite,Kaminski:2010qb}. The trace $\tr_{(k,\rho)}$ is over the infinite-dimensional Hilbert space $\ch_{(k,\rho)}$ that carries the principal-series unitary irreducible representation of $\Slc$ labelled by $k\in\Z_+$ and $\rho=\gamma (k+2)$, where $\g>0$ is the Barbero-Immirzi parameter. A canonical SU(2) subgroup has been chosen in $\Slc$ by the time gauge. The $\Slc$ representation can be decomposed into a direct sum of SU(2) irreducible representations: $\ch_{(k,\rho)}\cong \oplus_{k'=k}^\infty \ch_{k'}$. This allows us to define an orthonormal basis $|(k,\rho),k',m\rangle$ (for $m=-k/2,\cdots,k/2$), called the canonical basis of $\ch_{(k,\rho)}$. The projection operator $P_k$ associated to $\ch_k$ subspace \cite{hanPI,EPRL} is given by:
\be
P_k=\sum_{m=-k/2}^{k/2}\big| (k,\rho),k,m\big\rangle\big\langle (k,\rho),k,m\big|\ ,
\ee
where the state $|(k,\rho),k,m\rangle$ forms a basis in $\ch_k\subset\ch_{(k,\rho)}$.

In the formula \eqref{Avstackampliude}, the stack vertex amplitude $\sa_v$ is a function of half-link SU(2) holonomies $H_{ef}^{(i)}$ on the boundary, $i=1,\cdots,p_f$. The composition $H_{e' f}^{(i)}H_{ef}^{(i)-1}$ is the holonomy along the $i$-th link connecting the end points of $e,e'$. The reason why we split each link into halves is becoming clear in a moment.

Due to the sum over the face multiplicity $p_f$, the amplitude $\sa_v$ sums the spinfoam amplitudes on a series of different complexes. These complexes share the single vertex $v$ and edges $e$ but have different face multiplicities. The key point is that the spinfoam amplitude factorizes into contributions at individual faces under the $\Slc$ integrals, so the sum over $p_f$'s can be carried out independently at each root face $f$.

The neighborhood $U_v$ of the vertex $v$ is a 4-ball. We make the partition of the boundary $\partial U_v=S^3$ into polyhedra $R_e$ for each $e$ connecting at $v$, such that $\partial U_v=\cup_e R_e$. The intersection $R_e\cap R_{e'} $ is at their boundaries if $e,e'$ shares a root face, otherwise the intersection is empty. Each polyhedron $R_e$ encloses the node at $e\cup \partial U_v$. Each interface $\Fs_{f}=R_e\cap R_{e'} $ is dual to the stacked faces at $f$. The amplitude $\sa_v$ as a function of holonomies $H^{(i)}_{ef}$ is a spin-network stack on the boundary $\partial U_v$.

Intuitively, since the model is Lorentzian and all the regions in $\partial U_v\cong S^3$ are endowed with the spacelike quantum geometry (as SU(2) spin-networks), a subset of the polyhedra are in the causal future of the rest \footnote{Rigorously speaking, the subset of polyhedra in the causal future depends on the boundary state. Before taking the inner product between $\sa_v$ and boundary states, $\sa_v$ encodes all possible division of $\partial U_v$ into future and past.}. We denote by $B_+$ the union of these polyhedra and denote by $B_-$ the union of the polyhedra in the causal past. The interface between $B_+$ and $B_-$ corresponds to the corner of a causal diamond, see FIG.\ref{diamond} for an illustration. This vertex amplitude describes the local dynamics of spin-network stacks inside a causal diamond.

Note that in the formula \eqref{Av1}, the sum over $k_1,\cdots,k_p\in\Z_+$ is not constrained by the triangle inequality, because the triangle inequality of spins is imposed by the integrals over $g_{ve}$.

\begin{figure}[t]
\centering
\includegraphics[width=0.3\textwidth]{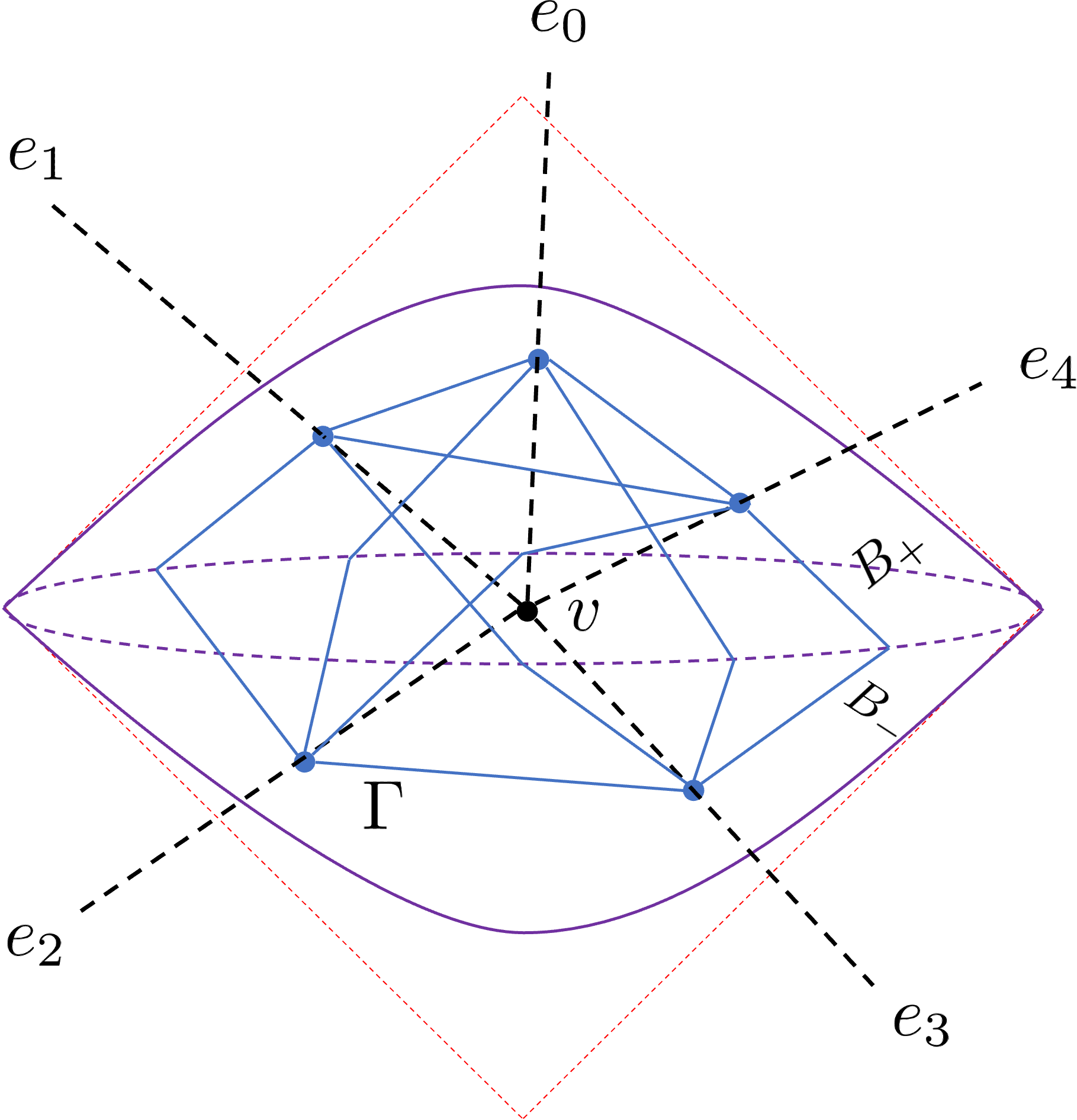}
\caption{The intuitive picture of the vertex amplitude $\sa_v$ for a valent-5 vertex $v$: $B_+\cup B_-=\partial B_4$ (drawn in purple lines) is the 3d boundary of a neighborhood of the vertex. $B_+$ is in the causal future of $B_-$. $B_+\cap B_-$ (the purple dashed circle) is the corner of the causal diamond (the red dashed diamond). The edges and faces of the spinfoam leaves the nodes and links (in blue) in $\partial B_4$. Here only the root graph is drawn, so each blue link corresponds to arbitrarily number of links for $\sa_v$. Some of the links are across the conner. Each nodes corresponds to a closed subregion such that $B_\pm$ are the unions of subregions.}
\label{diamond}
\end{figure}

\subsection{Stack amplitudes on arbitrary complex}\label{Stacking spinfoams}

The stack amplitude $\sa_\ck$ on any root 2-complex is given by the product of $\sa_v$ over all vertices followed by integrating over the boundary data $H^{(i)}_{ef}$ to glue the vertex amplitudes
\be
\sa_\ck\left(\vec{A},\vec{H},\vec{\l}\right)=\int [\rmd H] \prod_v\sa_v\lt(\vec{A},\vec{H},\vec{\eps}\rt).\label{innerprodKin}
\ee
The integral $\int [\rmd H]$ is the product Haar integration over all internal $H^{(i)}_{ef}$ for gluing the vertex amplitudes. The dimension of the integral is finite due to the cut-off $A_f$ that is identified between the vertex amplitudes sharing the same $f$. When a pair of vertex amplitudes $\sa_v,\sa_{v'}$ are glued, the integral picks up the corresponding terms with the same $p_f$ in $\sa_v$ and $\sa_{v'}$, since $\int dH\, D^k(H)=0$ for $k\neq 0$. The nontrivial integrals have the following pattern
\be
\int \rmd H_1\,d_{k'}\tr_{k'}({H}^{-1}_0 A H_1)\,d_k\tr_k(H_1^{-1}B H_2 )=d_k\delta^{kk'}\tr_k({H}^{-1}_0 AB {H}_2 ),
\ee
where $A,B$ stand for $P_{k_i}g_{ve}^{-1}g_{ve'}P_{k_i}$ at different vertices. This integral illustrated graphically in FIG.\ref{glueface} glues the faces in the vertex amplitudes.

\begin{figure}[h]
\centering
\includegraphics[width=0.8\textwidth]{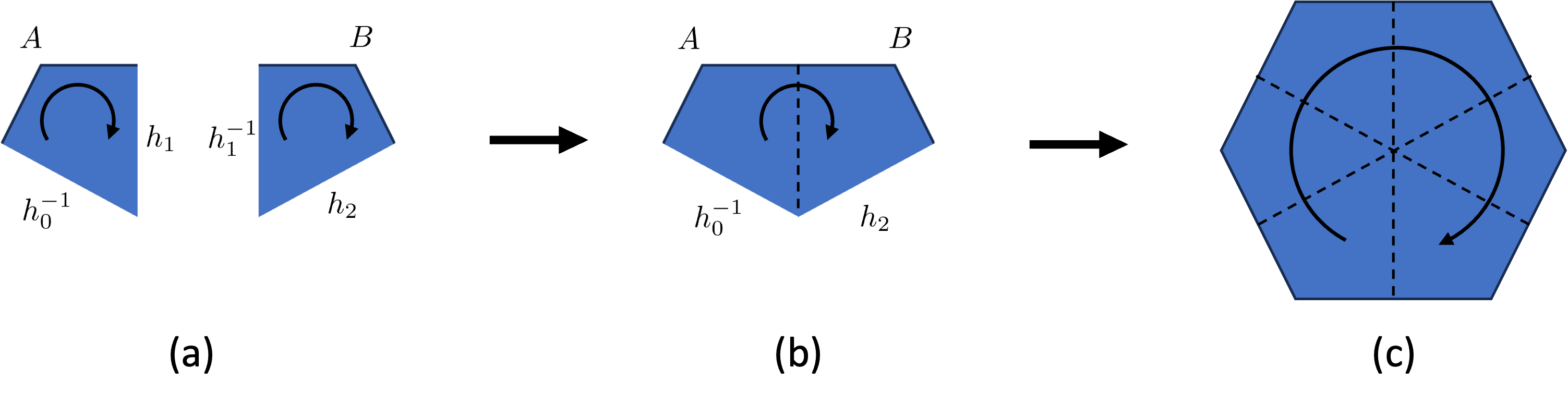}
\caption{(a)$\to$(b): Gluing a pair of faces in two different vertex amplitudes. (c) Generalization to gluing the faces in six vertex amplitudes to form an internal face.}
\label{glueface}
\end{figure}

The coupling constant $\l_f$ at the face $f$ in the root complex $\ck$ is a product of $\l_{vf}$ of the vertex amplitudes:
\be
\l_f=\prod_{v\in\partial f} \l_{vf}
\ee
If we let $\l_{vf}$ to be constant among $v\in\partial f$: $\l_{vf}\equiv e^{\phi_f}$, it enters the stack amplitude by
\be
\l_f^{p_f}=e^{\Phi_f p_f},\qquad \Phi_f=|V_f|\phi_f
\ee
where $|V_f|$ denotes the number of vertices on the boundary of $f$. $\Phi_f$ coupling to the multiplicity $p_f$ may be interpreted as a chemical potential for the stacked faces.

The stack amplitude $\sa_\ck$ is expressed by the following integral formulation:
\be
\sa_{{\cal K}}\left(\vec{A},\vec{\ff},\vec{\l}\right)&=&\int\prod_{(v,e)}\rmd g_{ve}\prod_{f=h,b}\o_f\left(A_f;\{g_{ve}\},\vec{\ff},\l_f\right)\label{saKampli}\\
\o_f&=&\sum_{p_f=1}^{\infty}\sum_{k_1,\cdots,k_{p_f}\in\Z_+}\prod_{i=1}^{p_f}\zeta^{(f)}_{k_{i}}\left(\{g_{ve}\},\ff^{(i)},\l_f\right)\Theta\left(A_{f}-\alpha_{p_f,\vec{k}}\right).\label{omegaf20}
\ee
We adopt a consistent notation where $v, e,$ and $f$ represent the vertices, edges, and faces of the root complex $\ck$. We distinguish internal faces, labelled as $h$, from boundary faces, labelled as $b$. The internal face does not connect to the boundary $\partial \ck$. All $p_f$ faces stacked at $f$ share the same boundary edges, so they all depend on the same set of group variables $g_{ve}\in\Slc$. This evolves from the fact that stacked links in a spin-network stack connect to the same pair of nodes. The stacked faces carry spins $k_i/2$, $i=1,\cdots,p_f$.

The sum over all 2-complexes in the spinfoam stack is encoded in Eq. \eqref{saKampli} through the summations over $p_f$ for each face $f$ in the root complex. Every complex in this sum is generated by stacking faces onto $\ck$. The sum over spins $k_i/2$ includes all faces, even boundary ones, compatible to the boundary spin-network stack states. To ensure the sums over both $p_f$ and their spins $k_i$ are finite, a cutoff $A_f$ is applied to each face $f \subset \ck$. In the expression \eqref{saKampli}, we have taken into account the boundary state labelled by $\vec{\ff}$. The functions $\o_b$ and $\zeta^{(b)}_{k_{i}}$ for boundary faces $f=b$ depend on the boundary state $\ff^{(i)}=\ff_2^{(i)}\otimes \ff_1^{(i)}\in \ch_{k_i}\otimes\ch_{k_i}^*$, where $i$ labels the stacked faces. In contrast, the functions $\o_h$ and $\zeta^{(h)}_{k_{i}}$ for internal faces $f=h$ are independent of $\vec\ff$. The explicit definition of $\zeta^{(f)}_k$ is provided below.

\begin{itemize}
\item Internal face $f=h$:
\be
\zeta_{k}^{(h)}=\l_h \tau_k^{(h)},\qquad \tau_k^{(h)}= d_k\tr_{(k,\rho)}\lt[\overrightarrow{\prod_{v\in\partial h}}P_kg_{ve}^{-1}g_{ve'}P_k\rt],\qquad \rho=\g (k+2)\label{zetakh}
\ee
where $d_k=k+1$ is the dimension of the spin-$k/2$ representation of SU(2). At any vertex $v$, the edges $e$ and $e'$ are, respectively, the incoming and outgoing edges according to the face's orientation.  

\item Boundary face $f=b$:
\be
\zeta_k^{(b)}=\l_b\tau_k^{(b)},\qquad \tau_k^{(b)}= d_k\lag \ff_1\lt|\overrightarrow{\prod_{v\in\partial b}}P_kg_{ve}^{-1}g_{ve'}P_k\rt|  \ff_2\rag. \label{zetakb}
\ee
This is defined for normalized boundary states $\ff_2\otimes \ff_1\in\ch_k\otimes\ch_k^*$, with the inner product taken in the Hilbert space $\ch_{(k,\rho)}$. For example, if coherent intertwiners are used for the boundary, $\ff_1$ and $\ff_2$ would be coherent states such as $\langle\ff^{(i)}_{1}|=\langle k_i,\xi^{(i)}_{e_{1} b}|$ and $|\ff^{(i)}_{2}\rangle=|k_i,\xi^{(i)}_{e_{2} b}\rangle$. The index $i=1,\cdots,p_b$ corresponds to the stacked links at the boundary, and $e_1, e_2$ are the edges that connect to the end points of the links. For the rest of this paper, we will leave $\vec \ff$ as arbitrary.

\end{itemize}


The integrand of $\sa_\ck$ is invariant under a set of continuous gauge transformations:
\be
g_{ve}\to x_v g_{ve} u_e,\qquad x_v\in\Slc,\quad u_e\in\Su.
\ee
The $\Slc$ gauge freedom leads to a divergence. As in the vertex amplitude, we address this by fixing the gauge, choosing a specific edge $e_0$ at each vertex $v$ and setting its corresponding variable to the identity, $g_{ve_0}=1$. 

Note that the stacked faces at a given internal face $h$ share the same boundary and thus creates some closed surfaces, or namely bubbles. The triangle inequality cannot give upper bound to spins on any bubble. The cut-off $A_h$ is the regulartor for the bubble divergence.

\section{Localization to the space of flat connections}\label{Localization to the space of flat connections}

Let us focus on $\o_h$ of an internal root face $f$ in the integrand of \eqref{saKampli}. The method of Laplace transform \cite{lqgee1,Agullo:2010zz,BarberoG:2008dwr} can be applied to $\o_h$ for computing the sum over the states associated to the stacked faces at $h$: Given two sequence $\{\a_n\}_{n=1}^\infty $ and $\{\b_n\}_{n=1} ^\infty $ where $\b_n\in\C$ and $\a_n>0$, if $\sum_{n=1}^\infty\left|\beta_{n}\right|e^{-\alpha_{n}\re(s)}<\infty$ for some $\re(s)>0$, we have
\be
\sum_{n,\alpha_{n}< A}\beta_{n}=\sum_{n=1}^{\infty}\beta_{n}\Theta\left(A-\alpha_{n}\right) =\frac{1}{2\pi i}\int_{T-i\infty}^{T+i\infty}\frac{\rmd s}{s}\left[\sum_{n=1}^{\infty}\beta_{n}e^{-\alpha_{n}s}\right]e^{As}.\label{invLaplace}
\ee
for the cut-off $A$ that does not coincide with any $\a_n$. The parameter $T>0$ is greater than the real part of all singularities given by the integrand.

We apply the formula \eqref{invLaplace} to the state-sum in $\o_h$ for internal face $h$: The sum over $n$ in \eqref{invLaplace} corresponds to the sum over $p_h$ and $k_1,\cdots,k_{p_h}$, and $\a_n$ corresponds to $\a_{p_h,\vec{k}}$. The summand $\b_n$ corresponds to $\prod_{i=1}^{p_h} \zeta_{k_i}^{(h)}(g_{ve},\l_h)$. A useful fact is that $\zeta_{k_i}^{(h)}(g_{ve},\l_h)$ depends on $i$ only through $k_i$. Applying \eqref{invLaplace} to $\o_h$ gives
\be
\o_h=\frac{1}{2\pi i}\int\limits_{T-i\infty}^{T+i\infty}\frac{\rmd s_h}{s_h}e^{A_h s_h}\sum_{p_h=1}^{\infty}\lt[\sum_{k=1}^\infty\zeta_k^{(h)} e^{-s_h\sqrt{k(k+2)}}\rt]^{p_h}=\frac{1}{2\pi i}\int\limits_{T-i\infty}^{T+i\infty}\frac{\rmd s_h}{s_h}e^{A_h s_h}\frac{\sum_{k=1}^\infty \zeta^{(h)}_k e^{- s_h\sqrt{k(k+2)}}}{1-\sum_{k=1}^\infty\zeta^{(h)}_k e^{- s_h\sqrt{k(k+2)}}}\ .\label{shintegral}
\ee
For a large cut-off $A_h$, the integral is dominated by the pole (in the $s_h$-plane) of the integrand with the largest $\re(s_h)$. The pole with $\re(s_h)>0$ can only be obtained by 
\be
\sum_{k=1}^\infty\zeta^{(h)}_k(g_h,\l_h)\, e^{- s_h(g_h,\l_h)\sqrt{k(k+2)}}=1.\label{poleeqn}
\ee
The solution $s_h(g_h,\l_h)$ depends on the group variables $g_h=\{g_{ve}\}_{e\subset\partial h}$ and the coupling constant $\l_h$. In order to obtain the maximum of $\re(s_h)$ among solutions, we use the following bound:

\begin{lemma}\label{zetabound}
$|\zeta_k^{(h)}|\leq \l_h d_k^2$, the equality holds if and only if $g_h$ satisfies
\be
g_{ve}^{-1}g_{ve'}\in \Su,\quad \forall e,e'\subset\partial h, \qquad \overrightarrow{\prod_{v\in\partial h}} g_{ve}^{-1}g_{ve'} =\pm1\ .\label{ghpm1}
\ee

\end{lemma}

\begin{proof}
See \cite{spinfoamstack}.
\end{proof}

Combining this bound and the equation \eqref{poleeqn} for the pole, we obtain the following result:

\begin{theorem}
Given any $\l_h>0$, $\re(s_h(g_h,\l_h))$ reaches the maximum if and only if $g_h$ satisfies
\be
g_{ve}^{-1}g_{ve'}\in \Su,\quad \forall e,e'\subset\partial h, \qquad \overrightarrow{\prod_{v\in\partial h}} g_{ve}^{-1}g_{ve'} =1\ .\label{g0solution} 
\ee

\end{theorem}

\begin{proof}
The upper bound of $\re(s_h)$ at the pole is derived from \eqref{poleeqn} by
\be
1\leq \sum_{k=1}^\infty\lt|\zeta^{(h)}_k(g_h,\l_h) \rt|e^{- \re\lt(s_h(g_h,\l_h)\rt)\sqrt{k(k+2)}}\leq \l_h\sum_{k=1}^\infty d_k^{2}e^{- \re(s_h(g_h,\l_h))\sqrt{k(k+2)}}.\label{bound28}
\ee
The right-hand side monotonically decreases as $\re(s)$ grows. Therefore, the solution $s_h(g_h,\l_h)$ of \eqref{poleeqn} satisfies $\re(s_h(g_h,\l_h))\leq \b_h(\l_h)$, where $\b_h(\l_h)>0$ satisfies 
\be
\l_h\sum_{k=1}^\infty d_k^{2}e^{- \b_h\sqrt{k(k+2)}}=1,\qquad 
\ee
When $\re(s_h(g_h,\l_h))$ reaches the maximum for some $g_h$: $\re(s_h(g_h,\l_h))=\b_h(\l_h)$, the inequality \eqref{bound28} implies 
\be
\sum_{k=1}^{\infty}|\zeta_{k}^{(h)}(g_h,\l_h)|e^{-\beta_{h}(\l_h)\sqrt{k(k+2)}}=1,
\ee
then the equality $|\zeta_{k}^{(h)}|=\lambda_{h}d_{k}^{2}$ in Lemma \ref{zetabound} must hold, and it restricts $g_h$ to satisfy \eqref{ghpm1}. In addition, the case with $\overrightarrow{\prod}_{v\in\partial h} g_{ve}^{-1}g_{ve'}= -1$ is ruled out, because in this case, Eq.\eqref{poleeqn} gives $\l_h\sum_{k=1}^\infty (-1)^k d_k^{2}e^{- s_h\sqrt{k(k+2)}}=1$,
which implies $\re(s_h)<\b_h$ strictly, see Appendix of \cite{spinfoamstack}.

Conversely, this restriction implies $\zeta_k^{(h)}=\l_h d_k^2$ and $s_h(g_h,\l_h)=\b_h(\l_h)$. 

\end{proof}


For large $A_h$, $\o_h$ is given by
\be
&&\o_h= e^{A_h s_h(g_h,\l_h)}\cf_h(g_h,\l_h) +\Fr_h ,\qquad \cf_h(g_h,\l_h)=\frac{1}{s_h(g_h,\l_h)\lag \sqrt{k(k+2)} \rag_{h,g}},\label{omegah31}\\
&&\lag \sqrt{k(k+2)} \rag_{h,g}=\l_h\sum_{k=1}^\infty \sqrt{k(k+2)} \t^{(h)}_k(g_h) e^{- s_h(g_h)\sqrt{k(k+2)}},
\ee
The maximum of $\re(s_h(g_h,\l_h))$ equals $\beta_h$. All other poles\footnote{Other possible poles includes $s_h=0$ and $s_h$ that solves \eqref{poleeqn} but cannot reach $\re(s_h)=\b_h$, although the solution of \eqref{poleeqn} is generally non-unique. } in \eqref{shintegral} have $\re(s_h)$ strictly less than $\b_h(\l_h)$. Their contributions collected by $\Fr_h$ are subleading in $\sa_\ck$ since $e^{A_h\re(s_h)}\ll e^{A_h\b_h}$ (see \cite{spinfoamstack} for some more discussion).

Applying \eqref{omegah31} to $\sa_\ck$ and neglecting $\Fr_h$, we obtain the following expression
\be
\sa_\ck&\simeq& e^{\sum_h \beta_h(\l_h) A_h}\int \prod_{(v,e)}\rmd g_{ve}\, e^{S(g,\l)}\,{\prod_h \cf_h(g_h,\l_h)}\prod_b\o_b\lt(A_b,\{g_{ve}\}_{e\subset\partial b},\vec\ff,\l_b\rt), \label{AKexponen0}\\
&&S(g,\l)=\sum_h A_h \lt[s_h\lt(g_h,\l_h\rt)-\b_h\lt(\l_h\rt)\rt] 
\ee
We scale uniformly $A_h\to\infty$ and apply the stationary phase approximation to the integral over the subspace containing $\{g_{ve}\}_{e\in E_{\rm int}}$, where $E_{\rm int}$ is the subset of the edges that does not connect to the boundary. We aim at an asymptotic expansion in $A_h$. The ``action'' $S$ relates only to internal faces $h$ and thus only depends on $\{g_{ve}\}_{e\in E_{\rm int}}$. The real part of $S$ reaches the maximum $\re(S)=0$ if and only if \eqref{g0solution} is satisfied by $g_h$ for all $h$, and it implies $S=0$. We denote by $\cc_{\rm int}$ the space of $\{g_{ve}\}_{e\in E_{\rm int}}$
satisfying \eqref{g0solution} for all internal faces $h$. A point in $\cc_{\rm int}$ is denoted by $g_{\rm 0,int}(\vec{u})$, where $\vec{u}$ parametrizes the solutions. $S$ is analytic in a neighborhood of $\cc_{\rm int}$ (see Appendix \ref{Analyticity of shg}).

The space $\cc_{\rm int}$ is the critical manifold of $S$: At any $g_{\rm 0,int}\in \cc_{\rm int}$
\be
\frac{\partial}{\partial{g_{ve}}}\sum_h A_h s_h\Big|_{g_{\rm 0,int}}=\sum_{h;e\subset\partial h}A_h\sum_{k_h=1}^\infty\frac{\partial}{\partial{g_{ve}}}\zeta^{(h)}_{k_h}\Big|_{g_{\rm 0,int}} e^{- \b_h\sqrt{k_h(k_h+2)}}=0,\label{EOMclos}
\ee
because the derivative of $\zeta^{(h)}_k$ at ${g_{\rm 0,int}}$ is proportional to $\tr_{(k,\rho)}[P_kJ^{IJ}P_k]$, which vanishes for all so(1,3) generator $J^{IJ}$. The exponent $S$ satisfies $\re(S)=\partial_g S=0$ only on $\cc_{\rm int}$. Therefore, the leading order of the asymptotic expansion as $A_h\to\infty$ localizes the integral onto the critical manifold $\cc_{\rm int}$.

The critical manifold $\cc_{\rm int}$ interestingly relates to the space of SU(2) flat connections on $\ck$: We have fixed the $\Slc$ gauge freedom in $\sa_\ck$ by choosing an edge $e_0(v)$ at every $v$ and setting $g_{ve_0(v)}=1$, and $e_0(v)\neq e_0(v')$ for $v\neq v'$. For some convenience becoming clear below, we require that $e_0(v)\in E_{\rm int}$, i.e.  $e_0(v)$ does not connect to $\partial\ck$, for all $v$. After the $\Slc$ gauge fixing, there is still the residue SU(2) gauge freedom:
\be
g_{ve}\to u_{e_0(v)}^{-1}g_{ve}u_e,\qquad u_e\in\Su,
\ee
which leaves the gauge fixing $g_{ve_0(v)}=1$ invariant. This SU(2) gauge transformation induces the on-shell gauge symmetry on $\cc_{\rm int}$. We denote by the on-shell gauge orbit by $\cg_{\rm int}$. The critical manifold $\cc_{\rm int}$ has the following property:

\begin{theorem}\label{SU2flatconn}
Assume the root 2-complex $\ck$ to be connected and sufficiently refined, 

(1) $\cc_{\rm int}/\cg_{\rm int}$ is identical to the moduli space $\cm(\ck)$ of SU(2) lattice flat connections. $\cm(\ck)$ is the space of $\{g_{e}\}_{e\in E_{\rm int}}$, $g_e\in\Su$, modulo gauge transformations at vertices and satisfying the flatness: the loop holonomy $\overrightarrow{\prod}_{e\subset\partial h}g_e$ is trivial around every internal face $h$. 

(2) The moduli space $\cm(\ck)$ is identical to the space of representations of the lattice fundamental group $\pi_1(|\ck|)$ in $\Su$ modulo conjugation:
\be
\cc_{\rm int}/\cg_{\rm int}\cong \mathcal{M}(\ck) \cong \mathrm{Hom}(\pi_1(|\ck|), \Su) / \Su.
\ee
where the lattice fundamental group $\pi_1(|\ck|)$ equals to $\pi_1(\mathrm{sk}(\ck))$ quotient by the normal subgroup generated by $\partial h$. When $\ck$ is embedded in a 4-manifold $\sm_4$, and $\ck $ is sufficiently refined such that $\pi_1(\sm_4)=\pi_1(|\ck|)$, The quotient space $\cc_{\rm int}/\cg_{\rm int}$ is identical to the moduli space of SU(2) flat connection on $\sm_4$:
\be
\cc_{\rm int}/\cg_{\rm int}\cong \mathrm{Hom}\lt(\pi_1(\sm_4),\Su\rt)/\Su.\label{CCequalHom}
\ee

\end{theorem}

\begin{proof}

Let us consider the first condition in \eqref{g0solution}. At any vertex $v$ and any internal face $h$ bounded by $v$ and $e_0$ another edge $e_1$ connecting $v$, $g_{ve_0}^{-1} g_{ve_1}\in\Su$ restricts $g_{ve_1}\in\Su$. For any other internal face $h'$ bounded by $v$ and $e_1$ and another edge $e_2$ connecting $v$, $g_{ve_1}^{-1} g_{ve_2}\in\Su$ restricts $g_{ve_2}\in\Su$. The restriction can propagate to all $e$ connecting to $v$ and thus gives $g_{ve}\in\Su$ for all $e\in E_{\rm int}$ \footnote{Consider the dual cellular complex $\ck^*$, where $v$ is dual to a 4-cell $v^*$, and $e$ is dual to a 3d polyhedron $e^*$ on the boundary of the 4-cell. $f\subset \ck$ bounded by $e,e'$ at $v$ is dual to a face $f^*$ shared by two 3d polyhedra $e^*,e'{}^*$. This argument of propagating restrictions would become invalid if $\partial v^*$ contained some boundary polyhedra in $\partial\ck^*$ such that their complement in $\partial v^*$ is disconnected. This obstruction clearly cannot happen for simplicial $\ck^*$. In general even for non-simplical $\ck^*$, this obstruction can be removed by refining $\ck^*$.
}. 
The SU(2) gauge transformation $u_e$ at every $e\neq e_0(v)$ further transforms one of $g_{ve}$ and $g_{v'e}$ to $1$. Therefore, there is an 1-to-1 correspondence between solutions in $\cc_{\rm int}/\cg_{\rm int}$ and SU(2) lattice flat connections $\{g_{e}\}_{e\in E_{\rm int}}$, modulo SU(2) gauge transformation $u_{e_0(v)}\equiv u_v$ at vertices
\be
g_{e}\to u_{s(e)}^{-1} g_e u_{t(e)}.
\ee
The flatness means that the loop holonomy $g_{\partial h}$ along $\partial h$ is trivial for every internal face $h$. We denote by $\cm(\ck)$ to be the space of all SU(2) lattice flat connections modulo gauge equivalence, and we have $\cc_{\rm int}/\cg_{\rm int}\cong \cm(\ck)$.

There is an bijection between $\cm(\ck)$ and $\mathrm{Hom}\lt(\pi_1(|\ck|),\Su\rt)$ quotient by SU(2) conjugation. The detailed proof of this statement is given in Appendix \ref{proofTheoremIII3}. When $\ck$ is embedded in $\sm_4$, the complex $\ck$ is assumed to be sufficiently refined, such that any loop $l\in\pi_1(\sm_4)$ is homotopic to a loop $l'$ that lies entirely on the 1-skeleton ${\rm sk}(\ck)$ of $\ck$ and $\pi_1(\sm_4)=\pi_1(|\ck|)$. As a result,
\be
\cc_{\rm int}/\cg_{\rm int}\cong \cm(\ck)\cong \mathrm{Hom}\lt(\pi_1(|\ck|),\Su\rt)/\Su\cong\mathrm{Hom}\lt(\pi_1(\sm_4),\Su\rt)/\Su. 
\ee
The space $\mathrm{Hom}\lt(\pi_1(\sm_4),\Su\rt)/\Su$ is identical to the moduli space of SU(2) flat connections on $\sm_4$.


\end{proof}

For example, $\cc_{\rm int}/\cg_{\rm int}$ is 3-dimensional for $\sm_4=I \times \mathbb{T}^3$, and $\cc_{\rm int}/\cg_{\rm int}$ is 0-dimensional for $\sm_4= I\times \mathbb{S}^3$, where $I\subset\R$ is a time interval. 

The leading asymptotic behavior of $\sa_{\ck}$ as $A_h\to\infty$ is given by \footnote{$\cf_h=\lt[\b_h\l_h\sum_{k=1}^\infty \sqrt{k(k+2)} d_k^2 e^{- \b_h\sqrt{k(k+2)}}\rt]^{-1}$ is constant on $\cc_{\rm int}$. }
\be
\sa_{\ck}
&=& e^{\sum_h \beta_h A_h}\bA^{-\mathscr{D}_{\rm int}}\int_{\cc_{\rm int}}\rmd \mu(\vec u) \int \prod_{(v,e_b)} \rmd g_{ve_b}\, {\prod_b\o_b|_{g_{\rm 0,int}(\vec u)}}\lt[\phi(\vec{u})+O(\bA^{-1})\rt] ,\label{AKexponen}
\ee
where $e_b$ denotes the edges connecting to the boundary and $g_{ve_b}\in\Slc$. The exponent $\mathscr{D}_{\rm int}>0$ is one half of the dimension of the Hessian matrix $\partial^2_g S|_{g_{\rm 0,int}(\vec u)}$, which we assume to be nondegenerate, and $\phi(\vec{u})$ relates to the determinant of the Hessian matrix. The nondegeneracy of Hessian matrix is proven in the case of trivial $\pi_1(|\ck|)$ in Section \ref{Nondegenerate Hessian matrix}. $\bar{A}$ is the mean value of $\{A_h\}_h$. The uniform scaling of $A_h$ corresponds to $A_h=\bA a_h$ and scaling $\bA\to \infty$.

This asymptotic behavior shows that the divergence of $\sa_\ck$ is due to the prefactor $e^{\sum_h \beta_h A_h}\bA^{-\mathscr{D}_{\rm int}}$. The parameters $\beta_h$ might be interpreted as a ``surface tension''. On the other hand, this behavior might interesting relate to the statistical mechanics of quantum geometry, because the key step \eqref{shintegral} is almost the same as the state-counting of LQG black hole entropy (see e.g. \cite{lqgee1,Agullo:2010zz}), except that $\zeta_k^{(h)}$ is generally complex, whereas the analog in state-counting is real and positive. This relation is also suggested by studying entanglement entropy in spinfoam theory \cite{spinfoamstack}.


Localizing the integral in $\sa_\ck$ onto the space of SU(2) flat connection makes \eqref{AKexponen} share the similarity with the SU(2) BF theory \cite{Ooguri:1992eb}, although \eqref{AKexponen} is equipped with a different integral measure. Indeed, we split the coordinates: $\vec u=(\vec r,\vec u')$ where $\vec u'$ are along the gauge directions on $\cg_{\rm int}$ and $\vec r$ are coordinates of the moduli space $\cm(\ck)$.  In the integrand, for any boundary face $b$, $\o_b$ is gauge invariant at $v$ disconnect to the boundary, and the SU(2) gauge freedom $u_v$ at $v$ connecting to the boundary by $e_b$ can be removed by $g_{v e_b}\to u_v^{-1} g_{v e_b}$ and the invariance of the Haar measure $\rmd g_{v e_b}$, so the boundary contribution only depends on $\vec{r}$. We introduce the notation
\be
\sa_{\G,\fs}\lt(\vec{r},\{A_b\},\{\l_b\},\vec \ff\rt):=\int \prod_{(v,e_b)} \rmd g_{ve_b}\, {\prod_b\o_b|_{g_{\rm 0,int}(\vec r)}}.
\ee
The label $\G$ is the boundary of $\ck$, and $\fs$ will be explained in a moment. Eq.\eqref{AKexponen} reduces to an integral of $\sa_{\G,\fs}$ over the moduli space of flat SU(2) lattice connections: Schematically, 
\be
\sa_{\ck}
&=& e^{\sum_h \beta_h A_h}\bA^{-\mathscr{D}_{\rm int}}\int_{\cc_{\rm int}}\rmd \mu(\vec r,\vec u') \lt[\phi(\vec r,\vec{u}')+O(\bA^{-1})\rt]\sa_{\G,\fs}(\vec r) \nonumber\\
&=&e^{\sum_h \beta_h A_h}\bA^{-\mathscr{D}_{\rm int}}\int_{\cm(\ck)}\rmd\rho(\vec r)\,\sa_{\G,\fs}(\vec r).
\ee
The measure $\rmd\rho(\vec r)$ is obtained by integrating out the gauge freedom $\vec u$ and possibly depend on the choice of $\ck$. Although $\sa_\ck$ reduces to the integral over the moduli space of flat connections, the measure is generally different from the path integral measure of BF theory. 

\section{Trivial topology and triangulation independence}\label{Trivial topology and triangulation independence}

In this section, we focus on the case of trivial $\pi_1(|\ck|)$. In this case, the space $\cc_{\rm int}/\cg_{\rm int}$ is 0-dimensional, so $\vec r$-coordinates disappear, and $\vec{u}$ only parametrizes the on-shell SU(2) gauge freedom. Therefore,
\be
\sa_{\G,\fs}=\sa_{\G,\fs}\lt(\{A_b\},\{\l_b\},\vec \ff\rt)=\int \prod_{(v,e_b)} \rmd g_{ve_b}\, \prod_b\mathring{\o}_{b,\pm},\label{AGamma}
\ee 
and it only depends on the following data relating to the boundary:
\begin{itemize}

\item The boundary root graph $\G=\partial\ck$.

\item The boundary state $\vec\ff$, boundary area cut-off $A_b\equiv A_\fl$, coupling constant $\l_b\equiv\l_\fl$, where $\fl=b\cap \G$ is the link of $\G$ along $\partial b$.

\item Assigning a sign for each link by $\fs: L(\G)\mapsto\Z_2$, where $L(\G)$ denotes the set of links in $\G$. The assignment is $\fs(\fl)=+$ if the number of vertices along $b$ is greater than 1, otherwise $\fs(\fl)=-$. 

\end{itemize}
In \eqref{AGamma}, $\mathring{\o}_b$ equals $\o_b$ evaluated at $g_{ve}=1$ for all $e\neq e_b$. Explicitly, for a boundary face $b$ with vertices labelled by $1,\cdots,l$, where the vertices $v_1$ and $v_l$ are connected to the boundary by the edges $e_b$ and $e_b'$ along the boundary of $b$, assuming $l\geq 2$
\be
\mathring{\o}_{b,+}=\sum_{p_b=1}^{\infty}\sum_{k_1,\cdots,k_{p_b}\in\Z_+}^{\infty}\prod_{i=1}^{p_b}\mathring{\zeta}^{(b,+)}_{k_{i}}\Theta\left(A_{b}-\alpha_{p_b,\vec{k}}\right)\qquad
\mathring{\zeta}^{(b,+)}_{k_i}= \l_b d_{k_i}\lag \ff^{(i)}_1\lt|P_kg_{v_1, e_b}^{-1}P_{k}g_{v_l,e_b'}P_k\rt|\ff^{(i)}_2\rag.\label{zetabkob1}
\ee
This result does not depend on the number $l$ of vertices of the face $b$, as far as $l\geq 2$. However, for the special case that $l=1$,  we have
\be
\mathring{\o}_{b,-}=\sum_{p_b=1}^{\infty}\sum_{k_1,\cdots,k_{p_b}\in\Z_+}^{\infty}\prod_{i=1}^{p_b}\mathring{\zeta}^{(b,-)}_{k_{i}}\Theta\left(A_{b}-\alpha_{p_b,\vec{k}}\right),\qquad
\mathring{\zeta}^{(b,-)}_{k_i}= \l_b d_{k_i}\lag \ff^{(i)}_1\lt|P_kg_{v_1, e_b}^{-1}g_{v_l,e_b'}P_k\rt|\ff^{(i)}_2\rag=\zeta^{(b)}_{k_i}, \label{zetabkob2}
\ee
so $\mathring{\o}_{b,-}=\o_b$ is not affected by the restriction, because in this case, $\zeta_{k}^{(b)}$ and $\o_b$ are independent of $g_{ve}$ for $e$ not connecting to boundary.

Insert this result into \eqref{AKexponen}, the dependence of $\sa_\ck$ on the bulk and boundary data factorizes
\be
\sa_{\ck}
&=&\sn_\ck \sa_{\G,\fs}\lt(\{A_b\},\{\l_b\},\vec \ff\rt),\qquad \sn_\ck= e^{\sum_h \beta_h A_h}\bA^{-\mathscr{D}_{\rm int}}\int_{\cc_{\rm int}}\rmd \mu(\vec u) \lt[\phi(\vec{u})+O(\bA^{-1})\rt] .\label{AKexponenSimp}
\ee
Here $\sn_\ck$ depends on the root 2-complex $\ck$ but is independent of the boundary data, so $\sn_\ck$ is just a normalization constant of the amplitude. The renormalized stack amplitude 
\be
\sa_{\ck,\mathrm{ren}}\equiv\sa_{\ck}/\sn_\ck=\sa_{\G,\fs}
\ee
only depends on the above data relating to the boundary but is independent of the bulk structure of $\ck$. In particular, $\sa_{\ck,\mathrm{ren}}$ is invariant under any refinement of $\ck$ that preserves $\G,\fs$ and the triviality of fundamental group.

Let us discuss the sum over root complexes: We first sum over root complexes $\ck(\fs)$ that shares the same boundary data $\G,\fs$: $\sa_\fs=\sum_{\ck(\fs)}a_{\ck(\fs)} \sa_{\ck(\fs)}$. The result also equals to $\sa_{\G,\fs}$ up to renormalization:
\be
\sa_{\fs,\mathrm{ren}}=\sa_\fs/\sn_\fs=\sa_{\G,\fs},\qquad \sn_\fs=\sum_{\ck(\fs)}a_{\ck(\fs)}\sn_{\ck(\fs)},\qquad a_{\ck(\fs)}\in\C.
\ee
Then up to normalization, the complete amplitude in \eqref{sa6} is generally a finite linear combination of $\sa_{\fs,\mathrm{ren}}$ over $\fs\in\Z_2^{|L(\G)|}$:
\be
\sa= \sum_{\fs\in \Z_2^{|L(\G)|}}b_\fs\, \sa_{\fs,\mathrm{ren}}=\sum_{\fs\in \Z_2^{|L(\G)|}}b_\fs\, \sa_{\G,\fs},\qquad b_\fs\in\C.
\ee
Here, all $\ck$ in the sum are assumed to have trivial $\pi_1(|\ck|)$. The sum over topologies is beyond our discussion in this paper.

Although the above discussion uses the stacked boundary state that leads to $\mathring{\o}_{b,\pm}$ as a sum over spins and multiplicity $p_b$, it is also valid for non-stacked boundary state. All above discussions including Eqs.\eqref{AKexponen} - \eqref{AGamma} and \eqref{AKexponenSimp} applies to any smooth function $\o_b(g_{ve})$. For instance, one may modify $\sa_\ck$ by removing the sums in $\o_b$ and restrict to the root graph $p_b=1$, then the result is given by the same restriction to $\mathring{\o}_{b,\pm}$ in \eqref{zetabkob1} and \eqref{zetabkob2}.

\section{Discussion}\label{Discussion}

The results in the last section are based on uniformly scaling the cut-offs $A_h\to\infty$, while keeping the boundary state fixed. This uniform scaling is motivated by relating $A_h$ to the cosmological constant $A_h\sim \ell_C^2/\ell_P^2$, where the cosmological constant is $\L=1/\ell_C^2$, as the areas should not exceed the maximal scale of the cosmological horizon \cite{HHKR,curvedMink}. Uniformly scaling the cut-offs $A_h\to\infty$ is equivalent to the limit of small cosmological constant: $\ell_C^2/\ell_P^2\to\infty$. This limit might relate to the UV limit, because when we zoom in to microscopic scales, any macroscopic curvature from a cosmological constant becomes negligible. Thus, the UV limit of quantum gravity should correspond to a regime where the cosmological constant is negligible.

A key mechanism driving our result is that in the limit $A_h\to \infty$, the integral over $\mathrm{SL}(2,\mathbb{C})$ group elements localizes onto the critical manifold $\cc_{\rm int}$, which is the space of SU(2) flat connections in the 4-dimensions. The localization drastically simplifies the dynamics, revealing a topological theory sharing similarities to $\mathrm{SU}(2)$ BF theory. This suggests that the physical quantitiess in this limit should depend only on the topology of the manifold, not its detailed geometric structure including the choice of triangulation. This topological nature is precisely why the final renormalized amplitude $\mathscr{A}_{\rm ren} = \mathscr{A}_{\Gamma,\fs}$ becomes independent of the bulk structure of 2-complex $\mathcal{K}$. The factor $\mathscr{N}_{\mathcal{K}}$ absorbs all the non-universal, triangulation-dependent parts of the amplitude.

The emergence of a triangulation-independent, topological theory in the limit suggests a fix point relating to the UV, similar to the Asymptotic Safety scenario of quantum gravity. In the context of spinfoams, refining the triangulation is analogous to probing smaller scales. The fact that the renormalized amplitude $\sa_\ck$ is invariant under bulk refinement of the complex $\ck$ suggests a scale-invariant fixed point. The amplitude becoming independent of the discrete triangulation is the spinfoam analog of scale-invariance. 

The localization of the spinfoam integral effectively breaks the gauge symmetry in the bulk from $\mathrm{SL}(2,\mathbb{C})$ to SU(2). The origin of this phenomenon lies in the simplicity constraint which results in the projector $P_k$ in e.g. \eqref{zetakh}. However, since the Haar integral of $g_{ve_b}\in\Slc$ relating to the boundary state is not affected, the Lorentz covariance discussed in \cite{Rovelli:2010ed} is still valid here\footnote{A boundary $\Slc$ transformation $\L$ changing the timelike normal of a boundary polyhedron (dual to $e_b$) leads to $P_k\to P_k'= \L^{-1} P_k \L$. One of $\L$ and $\L^{-1}$ is absorbed into the Haar integral $\int d g_{ve_b}$, while the other transforms the boundary SU(2) holonomy $H_\fl$ by $\L_{s(\fl)}H_{\fl}\L_{t(\fl)}^{-1}$.}.

We need to explain how our formalism connect to Infrared behavior of the theory: We expect that the infrared regime of the theory should correspond to a finite $A_h$. In this case, the stack amplitude $\sa_\ck$ recovers as an expansion in the coupling constant $\l_f$, where the leading order is the spinfoam amplitude on the root complex $\ck$. Moreover, the study of entanglement entropy in \cite{spinfoamstack} suggests $\l_f$ to relate to the Barbero-Immirzi parameter $\g$ and become small as in the small $\g$ regime. Therefore, the spinfoam amplitude on the root complex becomes dominant for finite $A_h$ and small $\g$. This connects to the existing semiclassical results of spinfoam e.g. \cite{Han:2023cen,lowE2,claudio,claudio1,propagator3,propagator2,propagator1,Alesci:2009ys}, which are based on the root complex (dual to simplicial complex) and relate to the small $\g$ regime. In addition, the regime of these semiclassical results is that both the internal spin cut-offs and boundary spins are scaled uniformly large but finite, and it is different from the limit here: $A_h\to\infty$ while keeping boundary state fixed.

\section{Explicit computations for trivial topology}\label{Some explicit computations}

In this section, we use some explicit parametrizations of group variables $g_{ve}$ to demonstrate the above general argument. The parametrizations is also useful for compute the Hessian matrix of $S$. All discussions in this section focus on the case that $\pi_1(|\ck|)$ is trivial. 

\subsection{Parametrizations}

Given the root complex $\ck$, we number the vertices by $v=v_i$, $i=1,\cdots, n$. For every edge $e=(i,j)$ for certain $i,j=1,\cdots,n$ ($i\neq j$), the pair of group variables $g_{v_ie} $ and $g_{v_j e}$ are re-labelled as $g_{ij}$ and $g_{ji}$. We use the following decomposition to parametrize each $g_{ij}\in\Slc$
\be
g_{ij}=u_{ij}e^{-ir_{ij}K^{3}}v_{ij},\qquad u_{ij}=e^{-i\psi_{ij}^{\prime}L^{3}}e^{-i\theta_{ij}^{\prime}L^{2}}\in\Su,\qquad v_{ij}=e^{-i\psi_{ij}L^{3}}e^{-i\theta_{ij}L^{2}}e^{-i\phi_{ij}L^{3}}\in\Su,
\ee
where $\psi_{ij}^{\prime},\theta_{ij}^{\prime},\psi_{ij},\theta_{ij},\phi_{ij}$ are Euler angles. To fix the $\Slc$ gauge freedom, we set $r_{ij}=0$ for $g_{ij}=g_{v e_0(v)}$. The Lie algebra generators relate to the Pauli matrices by $L^i=\frac{1}{2}\sigma^i,\ K^i=\frac{i}{2}\sigma^i$.

For any internal face $h$, we label the vertices of $h$ by $1,\cdots,m$,
\be
\zeta_{k}^{(h)}&=&\lambda_{h}d_{k}\mathrm{Tr}_{(k,\rho)}\left[g_{12}P_{k}g_{21}^{-1}g_{23}P_{k}g_{32}^{-1}\cdots g_{ij}P_{k}g_{ij}^{-1}\cdots g_{m1}P_{k}g_{1m}^{-1}\right]
\ee
where
\be
g_{ij}P_{k}g_{ji}^{-1}=\left(u_{ij}e^{-ir_{ij}K^{3}}v_{ij}\right)P_{k}\left(v_{ji}^{-1}e^{ir_{ji}K^{3}}u_{ji}^{-1}\right).
\ee
is the ``holonomy'' along the edge $(i,j)$.

Choosing a base vertex $v_*$ and a maximal spanning tree $\ct$ in the 1-skeleton of $\ck$. A maximal spanning tree is a subgraph that connects all vertices but contains no loops, and for any vertex $v$, there is a unique path $\calp_{v_*\to v}$ within the tree $\ct$ from $v_*$ to $v$. We make a change of variable
\be
v_{ij}v_{ji}^{-1}= u_{ij}^{-1} x_i^{-1} H_{ij}x_j u_{ji},\qquad x_i=\mathfrak{hol}({\calp_{v_*\to i}})\Big|_{r_{ij}=0},
\ee
where $H_{ij}=H_{ji}^{-1}$ and $\mathfrak{hol}(\calp)$ denotes the holonomy along the path $\calp$ made by $g_{ij}P_{k}g_{ji}^{-1}$. This is valid by the critical point condition \eqref{g0solution} that restricts
\be
r_{ij}\approx 0.\label{H=1}
\ee
Here and in the following, we use $\approx$ for the equality that holds only on the critical manifold $\cc_{\rm int}$. It implies
\be
g_{ij}P_{k}g_{ji}^{-1}\approx \begin{cases}
x_i^{-1} x_j,& (i,j)\in \ct\\
x_i^{-1} H_{ij}x_j,& (i,j)\not\in \ct
\end{cases}
\ee
Then the critical point condition \eqref{g0solution} further restricts
\be
H_{12}H_{23}\cdots H_{m1}\approx1.\label{HHH}
\ee
for any internal face $h$. The set of $\{H_{ij}\}_{(i,j)}$ satisfying the \eqref{HHH} defines a lattice flat connection under the ``tree gauge'' that $H_{ij}=1$ along $\ct$. The critical manifold $\cc_{\rm int}$ is the space of $\{H_{ij}\}_{(i,j)}$ and some on-shell gauge freedom in $u_{ij},v_{ij}$.

By Theorem \ref{SU2flatconn} and trivial $\pi_1(|\ck|)$, the SU(2) holonomies made by $H_{ij}$ are trivial for all loops in $\mathrm{sk}(\ck)$. It implies
\be
H_{ij}\approx 1,\qquad \forall (i,j),
\ee
because $\calp_{v_*\to i}\circ (i,j)\circ \calp_{j\to v_*}$ form a closed loop and all $H_{ij}$ along $\ct$ are trivial. In $\cc_{\rm int}$, all degrees of freedom of $H_{ij}$ are fixed, so $\cc_{\rm int}$ are completely parametrized by the on-shell gauge freedom $u_{ij},v_{ij}$.

For a boundary face $b$, whose vertices are labelled by $1,\cdots,l$ with $l\geq 2$
\be
\zeta^{(b)}_k=\l_b d_k\lag \ff_1\lt|P_kg_{1, e_b}^{-1}g_{12}P_{k}g_{21}^{-1}g_{23}P_{k}g_{32}^{-1}\cdots g_{l1}P_{k}g_{1l}^{-1} g_{l,e_b'}P_k\rt|\ff_2\rag 
\ee
where 
\be
g_{ij}P_{k}g_{ji}^{-1}=\lt(u_{ij}e^{-ir_{ij}K^{3}}u_{ij}^{-1}\right)P_k\lt(x_iH_{ij}x_j^{-1}\rt)P_k\left(u_{ji}e^{ir_{ji}K^{3}}u_{ji}^{-1}\rt).\label{gPgParameter}
\ee
By changing variables that leaves the Haar measure invariant
\be
g_{1, e_b}=x_1^{-1}g'_{1, e_b},\qquad g_{m, e_b'}=x_m^{-1}g'_{m, e_b'},
\ee
we obtain $\zeta^{(b)}_k$ on $\cc_{\rm int}$ 
\be
\zeta^{(b)}_k\approx \l_b d_k\lag \ff_1\lt|P_kg_{1, e_b}^{-1}P_{k}H_{12}H_{23}\cdots H_{l-1,l}P_{k} g_{l,e_b'}P_k\rt|\ff_2\rag.
\ee
is constant on $\cc_{\rm int}$ due to $H_{ij}=1$:
\be
\zeta^{(b)}_k\approx \l_b d_k\lag \ff_1\lt|P_kg_{1, e_b}^{-1}P_{k}g_{l,e_b'}P_k\rt|\ff_2\rag,
\ee
which reproduces \eqref{zetabkob1}.

\subsection{Nondegenerate Hessian matrix}\label{Nondegenerate Hessian matrix}

For the stationary phase analysis for \eqref{AKexponen0}, we split the integral of $\{g_{ve}\}_{e\in E_{\rm int}}$ into the directions along $\cc_{\rm int}$ and the transverse directions. Using the parametrization in \eqref{gPgParameter}, the transverse directions are parametrized by $r_{ij}$ and $H_{ij}$, while the integral along $\cc_{\rm int}$ is over the on-shell gauge freedom $u_{ij},v_{ij}$. The integral over on-shell gauge freedom gives $\int\rmd\mu(\vec u)\cdots$ in \eqref{AKexponen}. The stationary phase approximation is applied to the integral over the transverse directions. Correspondingly, the Hessian matrix of the action $S=\sum_h A_h[s_h(g)-\beta_h] $ is computed with respect to the transverse directions. We denote the Hessian matrix by $\mathbb{H}$ and use $\a,\b$ as the coordinate index for the transverse directions:
\be
\mathbb{H}_{\a\b}=\sum_{h}\frac{A_{h}}{\bA}\partial_{\a}\partial_{\b}s_{h}\left(g\right)\Big|_{\cc_{\rm int}},\qquad 
\partial_{\a}\partial_{\b}s_h\left(g\right)\Big|_{\cc_{\rm int}}=\frac{\sum_{k=1}^{\infty}\partial_{\a}\partial_{\b}\zeta^{(h)}_{k}\left(g\right)\Big|_{\cc_{\rm int}}e^{-\beta_h\sqrt{k(k+2)}}}{\sum_{k=1}^{\infty}\lambda_{h}d_{k}^{2}\sqrt{k(k+2)}e^{-\beta_h\sqrt{k(k+2)}}}
\ee
To simplify the formulae, we assume $A_h=A$ to be constant, then 
\be
\mathbb{H}_{\a\b}=\sum_{h}\lt(C_0^{(h)}\rt)^{-1}\sum_{k_h=1}^\infty\partial_{\a}\partial_{\b}\t_{k_{h}}^{(h)}\left(g\right)\Big|_{\cc_{\rm int}}e^{-\beta_h\sqrt{k_{h}(k_{h}+2)}},\qquad C_0^{(h)}=\sum_{k=1}^{\infty} d_{k}^{2}\sqrt{k(k+2)}e^{-\beta_h\sqrt{k(k+2)}}.
\ee
where $\tau_k^{(h)}=\zeta_k^{(h)}/\l_h$.

To compute the second derivatives of $\t_k^{(h)}$, it is convenient to expand
\be
e^{-ir_{ij}K^{3}}&=& 1-ir_{ij}K^{3}-\frac{1}{2}r_{ij}^{2}\left(K^{3}\right)^{2}+O(r^3),\\
H_{ij}&=&e^{i\sum_{a=1}^{3}t_{ij}^{a}L^{a}}=1-i\sum_{a=1}^{3}t_{ij}^{a}L^{a}-\frac{1}{2}\sum_{a,b=1}^{3}t_{ij}^{a}t_{ij}^{b}L^{a}L^{b}+O(t^3),
\ee
where $t^a_{ij}=-t^a_{ji}$. Here $H_{ij}$ are not along the maximal spanning tree. $r_{ij},t^a_{ij}$ are coordinates transverse to $\cc_{\rm int}$.

The resulting Hessian matrix $\mathbb{H}_{\a\b}$ is a polynomial of the Barbero-Immirzi parameter $\g$. But $\mathbb{H}_{\a\b}$ becomes simplified if we only focus on the leading order of small $\gamma$. In particular, due to the simplicity constraint,
\be
\langle j,m|P_{k}K^{3}P_{k}|j,n\rangle=-\gamma\langle j,m|L^{3}|j,n\rangle=O\left(\gamma\right),
\ee
the Hessian matrix is a direct sum of $r$-$r$ and $t$-$t$ blocks as $\g\to0$, due to
\be
\frac{\partial^{2}}{\partial r_{ij}\partial t_{mn}^{a}}\t_{k}^{(h)}\Big|_{\cc_{\rm int}}=O\left(\gamma\right).\label{rtoffdiagonal}
\ee

Furthermore, the $r$-$r$ block is block-diagonal, where each small block associates to $r_{ij}$ at a given vertex $i$. Indeed, at the vertex $i$, we have the diagonal entries
\be
&&\sum_{h}\lt(C_0^{(h)}\rt)^{-1}\sum_{k_h=1}^\infty\frac{\partial^{2}}{\partial r_{ij}^{2}}\t_{k_{h}}^{(h)}\Big|_{\cc_{\rm int}}e^{-\beta_h\sqrt{k_h(k_h+2)}}\nonumber\\
&=&-\sum_{h;(i,j)\subset\partial h}\lt(C_0^{(h)}\rt)^{-1}\sum_{k_h=1}^\infty d_{k_{h}}\mathrm{Tr}_{(k_{h},\rho_{h})}\left[P_{k_{h}}\left(K^{3}K^{3}\right)P_{k_{h}}\right]e^{-\beta_h\sqrt{k_h(k_h+2)}}\nonumber\\
& =&-\sum_{h;(i,j)\subset\partial h}C_1^{(h)}+O\left(\gamma\right),\qquad \qquad C_1^{(h)}=\lt(C_0^{(h)}\rt)^{-1}\sum_{k=1}^\infty\frac{1}{6}d_k^{2}\left(d_k+1\right)e^{-\beta_h\sqrt{k(k+2)}},\label{rrdiagonal}
\ee
and the off-diagonal entries
\be
&&\sum_{h}\lt(C_0^{(h)}\rt)^{-1}\sum_{k_h=1}^\infty\frac{\partial^{2}}{\partial r_{ij}\partial r_{im}}\t_{k_{h}}^{(h)}\Big|_{\cc_{\rm int}}e^{-\beta_h\sqrt{k_h(k_h+2)}}\nonumber\\
&=&-\sum_{h;(i,j),(i,m)\subset\partial h}\lt(C_0^{(h)}\rt)^{-1}\sum_{k_h=1}^\infty d_{k_{h}}\mathrm{Tr}_{(k_{h},\rho_{h})}\left[P_{k_{h}}\left(K^{3}u_{jim}K^{3}\right)P_{k_{h}}u_{jim}^{-1}\right]e^{-\beta_h\sqrt{k_h(k_h+2)}}\nonumber\\
&=&-\sum_{h;(i,j),(i,m)\subset\partial h}C_1^{(h)}\cos\lt(\theta_{jim}\rt)+O\left(\gamma\right).\label{rroffdiagonal}
\ee
Here $\theta_{jim}$ is one of the Euler angles of $u_{jim}\in\Su$ containing the on-shell gauge freedom. For $r_{im}$ and $r_{jn}$ associate to two different vertices $i,j$
\be
\frac{\partial^{2}}{\partial r_{im}\partial r_{jn}}\t_{k_{h}}^{(h)}\Big|_{\cc_{\rm int}}=O(\g).
\ee
for any $m,n$. So the $r$-$r$ block is block-diagonal as $\g\to0$. 

\begin{lemma}
The $r$-$r$ block is nondegenerate as $\g\to0$.
\end{lemma}

\begin{proof}
The $r$-$r$ block is nondegenerate if and only if every small blocks associated to a vertex is nondegenerate. Focus on a single vertex $v$ and define a quadratic form $Q(r)=\sum_{e,e'}\mathbf{M}_{e,e'}r_e r_{e'}$ with $r_e,r_{e'}\in\R$, where $e,e'$ are edges connecting to $v$ but  not connecting to the boundary. The matrix $\mathbf{M}$ has the diagonals $\mathbf{M}_{e,e} = -\sum_{h; e\subset \partial h} C^{(h)}_1$ and off diagonals $\mathbf{M}_{e,e'} = -\sum_{h; e,e'\subset \partial h} C^{(h)}_1\cos(\theta_{e,e'})$. The quadraic can be written as
\be
Q(r)=-\sum_{h,v\in \partial h} C^{(h)}_1 T_h(r),\qquad T_h(r)=r_{e_1(h)}^2+r_{e_2(h)}^2+2r_{e_1(h)} r_{e_2(h)}\cos(\theta_{e,e'}),
\ee
where $e_1(h),e_2(h)\subset\partial h$ are the pair of edges connecting the vertex $v$. For any $r_e\in\R$, we have $T_h(r)\geq 0$, and it implies $Q(r)\leq 0$. Moreover, $Q(r)= 0$ if and only if $T_{h}(r)=0$. By gauge fixing $r_{e_0}=0$ for one edge $e_0$, $T_{h}(r)=0$ implies $r_e=0$ for all $e$ sharing an internal face $h$ with $e_0$. By induction, $r_e=0$ propagates to all edges $e$
\footnote{In the same way as the proof of Theorem \ref{SU2flatconn}, this argument would become invalid if $\partial v^*$ contained some boundary polyhedra in $\partial\ck^*$ such that their complement in $\partial v^*$ is disconnected. But if we assume $\ck^*$ is either simplicial or sufficiently refined, this obstruction cannot happen.
}. 
That $Q(r)=0$ implies $r=0$ under gauge fixing indicates that the small block associated to $v$ is nondegenerate, so the $r$-$r$ block is nondegenerate.
\end{proof}


For the $t$-$t$ block of the Hessian, we have the diagonal entries
\be
&&\sum_{h}\lt(C_0^{(h)}\rt)^{-1}\sum_{k_h=1}^\infty\frac{\partial^{2}}{\partial t_{ij}^{a}\partial t_{ij}^{b}}\t_{k_{h}}^{(h)}\Big|_{C_{\rm int}}e^{-\beta_h\sqrt{k_h(k_h+2)}}\nonumber\\
&=&-\sum_{h;(i,j)\subset\partial h}\lt(C_0^{(h)}\rt)^{-1}\sum_{k_h=1}^\infty d_{k_{h}}\mathrm{Tr}_{k_{h}}\left[L^{a}L^{b}\right]e^{-\beta_h\sqrt{k_h(k_h+2)}}\nonumber\\
&=&-\sum_{h;(i,j)\subset\partial h}C_2^{(h)}\delta^{ab},\qquad 
C^{(h)}_2=\lt(C_0^{(h)}\rt)^{-1}\sum_{k=1}^\infty \frac{1}{12}d_k^{2}(d_k-1)\left(d_k+1\right)e^{-\beta_h\sqrt{k(k+2)}}
\ee
and the off-diagonal entries
\be
&&\sum_{h}\lt(C_0^{(h)}\rt)^{-1}\sum_{k_h=1}^\infty\frac{\partial^{2}}{\partial t_{ij}^{a}\partial t_{mn}^{b}}\t_{k_{h}}^{(h)}\Big|_{\cc_{\rm int}}e^{-\beta_h\sqrt{k_h(k_h+2)}}\nonumber\\
&=&-\sum_{h;(i,j),(m,n)\subset\partial h}s_{ij}(h)s_{mn}(h)\lt(C_0^{(h)}\rt)^{-1}\sum_{k_h=1}^\infty d_{k_{h}}\mathrm{Tr}_{k_{h}}\left[L^{a}L^{b}\right]e^{-\beta_h\sqrt{k_h(k_h+2)}}\nonumber\\
&=&-\sum_{h;(i,j),(m,n)\subset\partial h}s_{ij}(h)s_{mn}(h)C_2^{(h)}\delta^{ab},
\ee
where $(i,j)$ and $(m,n)$ are not along the maximal spanning tree $\ct$. The sign $s_{ij}(h)=1$ if the orientation of the edge $(i,j)$ (for defining $H_{ij}$) aligns with the orientation of $\partial h$, otherwise $s_{ij}(h)=-1$. Unlike $r_{ij}$ which only associates to the vertex $i$, $t_{ij}$ relates to both vertices $i$ and $j$. So, the $t$-$t$ block of the Hessian is not block diagonal. The $t$-$t$ block only relates to the topological properties of the root complex $\ck$.

\begin{lemma}

The $t$-$t$ block is non-degenerate.

\end{lemma}

\begin{proof}

The $t$-$t$ block can be expressed as $\mathbf{H}_{tt}\otimes \bm{1}_{3\times 3}$. We define the quadratic form $Q(x)=x^T\mathbf{H}_{tt} x$, where $x$ is the real vector $x=(x_e)_{e\in E_{\rm NT}}$ with $E_{\rm NT}$ being the set of edges not along $\ct$. This quadratic form can be written as
\be
Q(x)=-\sum_h C_2^{(h)}\lt(\sum_{e\in E_{\rm NT},e\subset\partial h}s_e(h)x_e\rt)^2.
\ee
Since $C_2^{(h)}>0$, the quadratic form the negative semi-definite $Q(x)\leq0$. To prove that $\mathbf{H}_{tt}$ is non-degenerate, we need to show that $Q(x)=0$ if and only if $x=0$: The condition $Q(x)=0$ implies
\be
\sum_{e\in E_{\rm NT},e\subset\partial h}s_e(h)x_e=0,\qquad \forall h.
\ee
We interpret these equations in the language of algebraic topology. We define a 1-cochain $\a\in C^1(\ck,\R)$ by $\a(e)=x_e$ for $e\in E_{\rm NT}$ and $\a(e)=0$ for $e\subset\ct$. The above equations is precisely the condition that $\a$ is cocycle: $\delta_1\a =0$ where $\delta_1$ is the coboundary operator. If the cohomology group $H^1(\ck,\R)$ is trivial, we have $\a =\delta_0 \beta$ for some $\b\in C^0(\ck,\R)$, i.e. $\a(e)=\beta(j)-\b(i)$ for $e=(i,j)$. Then $\a(e)=0$ for $e\subset\ct$ implies $\b$ is constant on all vertices of $\ck$. Therefore, we obtain $\a=0$ and thus $x=0$.

The trivial cohomology group $H^1(\ck,\R)$ is a consequence from the trivial $\pi_1(|\ck|)$, because a trivial $\pi_1(|\ck|)$ implies the trivial $H_1(\ck,\Z)$ and thus trivial $H_1(\ck,\R)$ and $H^1(\ck,\R)=\mathrm{Hom}(H_1(\ck,\R),\R)$.

\end{proof}

At the leading order of small $\g$, the $t$-$t$ block is a constant on the critical manifold $\cc_{\rm int}$, while the $r$-$r$ block is not constant due to the dependence on $\theta_{jim}$. Given that the Hessian matrix is non-degenerate in the limit $\g\to0$, it is still non-degenerate for a generic value of $\g$, in particular for small $\g$.

Let us consider two simple examples: The first example uses the root complex $\ck=\Delta_3^*$, which is made by three vertices and a single triangular internal face $h$ \cite{Han:2021kll}. There are 3 edges along $\partial h$, and 2 edges are in the spanning tree $\ct$, so there is only one $H_{ij}$ not along $\ct$. The $t$-$t$ block is $3\times 3$ given by $C_2^{(h)}\bm{1}_{3\times 3}$.

As the second example, let us consider the root complex $\ck=\sig_{\text 1-5}^*$ being the 2-complex dual to the 1-5 pachner move of 4-simplex \cite{Han:2021kll}. The complex $\sig_{\text 1-5}^*$ has five vertices and ten triangular faces $h$. The spanning tree $\ct$ contains four edges connecting the vertex $1$ to other four vertices. The $t$-$t$ block is $18\times 18$ non-degenerate matrix given by
\be
C_2\left(
\begin{array}{cccccc}
 -3 & 1 & 1 & -1 & -1 & 0 \\
 1 & -3 & 1 & 1 & 0 & -1 \\
 1 & 1 & -3 & 0 & 1 & 1 \\
 -1 & 1 & 0 & -3 & 1 & -1 \\
 -1 & 0 & 1 & 1 & -3 & 1 \\
 0 & -1 & 1 & -1 & 1 & -3 \\
\end{array}
\right)\otimes \bm{1}_{3\times 3}.
\ee
We have assumed that $\l_h=\l$ was constant among all $h$, then $\b_h=\b$ and $C_2^{(h)}=C_2$ were also constant.

\begin{acknowledgements}

The author receives supports from the National Science Foundation through grants PHY-2207763 and PHY-2512890. 

\end{acknowledgements}

\appendix


\section{Analyticity of $s_h(g)$}\label{Analyticity of shg}

Viewing $\Slc$ as a real manifold, $s_h(g)$ has the following analytic property
\begin{lemma}

$s_h(g)$ is analytic in a neighborhood of the solution to \eqref{g0solution}.

\end{lemma}

\begin{proof} First of all, the canonical basis vector $|(k,\rho), k ',m\rangle$ is $K$-finite, where $K=\Su$ is the maximal compact subgroup of $\Slc$, then Harish-Chandra's analyticity theorem\footnote{Harish-Chandra's analyticity theorem states that if $V$ carries a unitary irrep of a semisimple Lie group $G$ with maximal compact subgroup $K$ then every $K$-finite vector is a weakly analytic vector. A vector $v\in V$ is $K$-finite if it is contained in a finite-dimensional subrepresentation of $K$. A vector $v\in V$ is weakly analytic vector if $\langle u |g| v\rangle$ ($g\in G$) is analytic on $G$ (as a real manifold) for any $u\in V$.} implies that it is weakly analytic, so the Wigner $D$-function of the $\Slc$ unitary irrep $D^{(k,\rho)}_{k_1 m_1,k_2 m_2}(g)=\langle (k,\rho), k_1,m_1 | g | (k,\rho), k_2,m_2\rangle$ is an analytic function on $\Slc$. Consequently, $\zeta_k^{(h)}(g)$ is an analytic function of $g_{ve}$'s for each $k$, since it is a polynomial of the $D$-functions.

Each term in the sum $\sum_{k=1}^\infty\zeta^{(h)}_k(g) e^{- s_h\sqrt{k(k+2)}}$ is analytic in $g_{ve}$ and $s_h$. The sum converges uniformly for $\re(s_h)>0$ by the bound $\sum_{k=1}^\infty\lt|\zeta^{(h)}_k(g) e^{- s_h\sqrt{k(k+2)}}\rt|\leq \l_h\sum_{k=1}^\infty d_k^{2}e^{- \re(s_h)\sqrt{k(k+2)}}\leq \l_h\sum_{k=1}^\infty d_k^{2}e^{- r_0\sqrt{k(k+2)}}<\infty$ for $\re(s_h)\geq r_0>0$. Therefore $F(g,s)\equiv \sum_{k=1}^\infty\zeta^{(h)}_k(g) e^{- s_h\sqrt{k(k+2)}}-1$ is an analytic function of $g_{ve}$ and $s_h$.

We denote by $g_0$ a solution to \eqref{g0solution}. $\partial_s F(g_0,\b_h)=-\l_h\sum_{k=1}^\infty  d^2_k \sqrt{k(k+2)} e^{- \b_h\sqrt{k(k+2)}}\neq 0$. By the analytic implicit function theorem, there exists a unique, analytic function $s_h(g)$ defined in an open neighborhood of $g_0$ that satisfies $F(g,s_h)=0$.

\end{proof}

\section{Complete the proof of Theorem \ref{SU2flatconn}}\label{proofTheoremIII3}

The purpose of this appendix is to prove
\be
\mathcal{M}(\ck) \cong \text{Hom}(\pi_1(|\ck|), \text{SU(2)}) / \text{SU(2)}.
\ee

We define a map $\Phi: \mathcal{M}(\ck) \to \text{Hom}(\pi_1(|\ck|), \text{SU(2)}) / \text{SU(2)}$: For a gauge equivalence class $[g]\in\cm(\ck)$, pick a representative $\{g_e\}$ and a base vertex $v_0$. This defines a homomorphism $\rho_{\{g_e\}}: \pi_1(|\ck|) \to \text{SU(2)}$ by the loop holonomies made by $\{g_e\}$. This is well-defined due to the flatness of $[g]$ and $\pi_1(|\ck|)=\pi_1(\mathrm{sk}(\ck))/N_{\rm int}$ where $N_{\rm int}$ is the normal subgroup generated by all loops around internal faces. A different representative $\{g'_e\}$ related by a gauge transformation $\{h_v\}$ yields a conjugate representation $\rho_{\{g'_e\}} = h_{v_0} \rho_{\{g_e\}} h_{v_0}^{-1}$. Thus, $\Phi([g]) = [\rho_{\{g_e\}}]$ maps to the space of conjugacy classes.

To show $\Phi$ is an isomorphism, we construct its inverse $\Psi: \text{Hom}(\pi_1(|\ck|), \text{SU(2)}) / \text{SU(2)}\to \mathcal{M}(\ck)$: Let $[\rho]$ be a conjugacy class of representations. Pick a representative $\rho: \pi_1(|\ck|) \to \text{SU(2)}$. We construct a flat lattice connection $\{g_e^\rho\}$: (1) Choose a maximal spanning tree $\ct$ in the 1-skeleton of $\ck$. (2) For each edge $e \in \ct$, set $g_e^\rho = I$. (3)  For each edge $e=(u,v) \notin \ct$, there is a unique fundamental loop $\ell_e$ based at $v_0$ formed by the path $P_{v_0,u}$ in $T$, the edge $e$, and the path $P_{v,v_0}$ in $T$. All other loops in $\mathrm{sk}(\ck)$ are products of fundamental loops. Define 
\be
g_e^\rho = \rho([\ell_e]).\label{treegauge}
\ee
This assignment $\{g_e^\rho\}$ is flat. For any internal face $h$, the homotopy class $[\partial h]$ is the identity element in $\pi_1(|\ck|)$. Therefore, the holonomy around the face is $\rho([\partial h]) = I$, satisfying the flatness condition.

The map $\Psi([\rho]) = [\{g_e^\rho\}]$ is well-defined: If we choose a different representative $\rho' = h \rho h^{-1}$, the new connection is $\{g_e^{\rho'}\}$. For $e \in \ct$, $g_e^{\rho'} = I$. For $e \notin \ct$, $g_e^{\rho'} = \rho'([\ell_e]) = h \rho([\ell_e]) h^{-1} = h g_e^\rho h^{-1}$. Consider the constant gauge transformation $h_v = h$ for all vertex $v \in V$ (we denote by $V$ the set of vertices in $\ck$). The transformed connection is $h_{t(e)} g_e^\rho h_{s(e)}^{-1} = h g_e^\rho h^{-1}$. This equals $g_e^{\rho'}$ for $e \notin \ct$, and for $e \in \ct$, $h I h^{-1} = I = g_e^{\rho'}$. Thus, $\{g_e^{\rho'}\}$ is gauge-equivalent to $\{g_e^\rho\}$.

Finally, we show $\Phi$ and $\Psi$ are inverses. $\Phi(\Psi([\rho])) = [\rho]$ by construction. Then we must show that for any $[g]\in\cm(\ck)$, $\Psi(\Phi([g])) = [g]$. Let $\{g_e\}$ be a representative of $[g]$. Let $\rho = \Phi([\{g_e\}])$ be its holonomy representation. Let $\{g'_e\} = \Psi([\rho])$ be the connection constructed from $\rho$ in the "tree gauge" (where $g'_e=I$ for $e \in T$) by \eqref{treegauge}. We need to show that $\{g_e\}$ is gauge-equivalent to $\{g'_e\}$.

Let's construct the required gauge transformation $h: V \to \text{SU(2)}$. For the base vertex $v_0$, set $h_{v_0} = I$. For any other vertex $v$, let $P_{v_0,v}$ be the unique path in the spanning tree $T$ from $v_0$ to $v$. Define $h_v ^{-1}= \mathrm{Hol}_{\{g_e\}}(P_{v_0,v})\equiv \mathrm{Hol}(P_{v_0,v})$ being the holonomy along the path $P_{v_0,v}$ made by $g_e$. Let $\{g''_e\}$ be the result of applying this gauge transformation to $\{g_e\}$: $g''_e = h_{t(e)} g_e h_{s(e)}^{-1}$. We show $g''_e = g'_e$ for all $e \in E$:

\begin{itemize}
\item $e=(u,v) \in \ct$. We have $\mathrm{Hol}(P_{v_0,v}) = g_e \mathrm{Hol}(P_{v_0,u})$, so $h_v^{-1} = g_e h_u^{-1}$. The gauge-transformed holonomy is $g''_e = h_v g_e h_u^{-1} = (h_u g_e^{-1}) g_e h_u^{-1} = I$. This matches $g'_e = I$.

\item $e=(u,v) \notin T$. The constructed holonomy is $g'_e = \rho([\ell_e])$. The loop is $\ell_e = P_{v,v_0} \circ e \circ P_{v_0,u}$. The holonomy of this loop in the original configuration $\{g_e\}$ is $\rho([\ell_e]) = {\rm Hol}(P_{v,v_0}) g_e {\rm Hol}(P_{v_0,u})$. We have ${\rm Hol}(P_{v_0,u}) = h_u^{-1}$ and $\mathrm{Hol}(P_{v,v_0}) = h_v$. By \eqref{treegauge}, we obtain $g'_e = h_v g_e h_u^{-1}=g''_e$.
\end{itemize}
Since $g''_e = g'_e$ for all edges $e$, the configuration $\{g'_e\}$ is gauge-equivalent to $\{g_e\}$. Thus $\Psi(\Phi([g])) = [g]$, and the isomorphism is proven.

\bibliography{muxin.bib}

\end{document}